\documentclass[10pt]{article}

\usepackage[paper=a4paper,left=29mm,right=29mm,top=28mm, bottom=28mm]{geometry}

\usepackage{hyperref}

\usepackage{academicons}

\usepackage{xcolor}

\usepackage{authblk}
\usepackage{tikz}
\usepackage{xfrac}
\usepackage{tabularx, multirow}
\usepackage{graphicx, tikz}
\usetikzlibrary{arrows.meta}
\usetikzlibrary{arrows,calc}
\usetikzlibrary{shapes}

\usepackage{verbatim}
\usepackage{booktabs}
\usepackage{comment}
\usepackage{ziffer}
\usepackage{chngcntr}
\usepackage[noend]{algpseudocode}
\usepackage{algorithm}

\usepackage{latexsym,amssymb,amsfonts,graphicx,makeidx,paralist}
\usepackage{theorem,rotating,ifthen,amsmath,subfigure,epsfig,nicefrac}
\usepackage{booktabs,framed}

\usepackage{color,framed, float,wasysym,cancel}

\colorlet{shadecolor}{gray!12}

\newcommand{\RR}{{\mathbb R}}

\newenvironment{desctight}
  {\begin{list}{}{\setlength\labelwidth{0pt}
        \setlength{\itemsep}{0.5pt}
        \setlength{\parsep}{0pt}
        \setlength\itemindent{-\leftmargin}
        }}
    {\end{list}}

\newtheorem{theorem}{Theorem}[section]
\newtheorem{example}[theorem]{Example}
\newtheorem{corollary}[theorem]{Corollary}
\newtheorem{lemma}[theorem]{Lemma}

\newtheorem{definition}[theorem]{Definition}

\newtheorem{remark}[theorem]{Remark}

\newenvironment{proof}{\noindent{\bf Proof~}}{\null\hfill $\Box$\par\medskip}

\newcommand{\MSOA}{\text{MSO}_1}

\newcommand{\w} {\mbox{W}}

\newcommand{\s} {\mathfrak{s}}

\begin{document}

%\title{Computing the dichromatic number on digraphs of bounded directed clique-width}

\title{Computing Directed Steiner Path Covers\thanks{A short version of 
this paper appeared in the Proceedings of the {\em 46th International Conference on Current Trends in Theory and Practice of Computer Science} (SOFSEM 2020), see \cite{GHKRRW20}.}}

\author[1]{Frank Gurski}
\author[1]{Dominique Komander}
\author[1]{Carolin Rehs}

\author[2]{Jochen Rethmann}

\author[1]{Egon Wanke}
\affil[1]{\small University of  D\"usseldorf,
Institute of Computer Science, 
40225 D\"usseldorf, Germany}

\affil[2]{\small Niederrhein University of Applied Sciences,
Faculty of Electrical Engineering and Computer Science,
47805 Krefeld, Germany}

\maketitle

%%%%%%%%%%%%%%%%%%%%%%%%%%%%%%%%%%%%%%%%%%%%%%%%%%%%%%%%%%%%%%%%%%%%%%%%%%%
%%%%%%%%%%%%%%%%%%%%%%%%%%%%%%%%%%%%%%%%%%%%%%%%%%%%%%%%%%%%%%%%%%%%%%%%%%%
%%%%%%%%%%%%%%%%%%%%%%%%%%%%%%%%%%%%%%%%%%%%%%%%%%%%%%%%%%%%%%%%%%%%%%%%%%%

\begin{abstract}
In this article we consider the {\sc Directed Steiner Path Cover} problem
on directed co-graphs. Given a directed graph $G=(V,E)$ and a set
$T \subseteq V$ of so-called terminal vertices, the problem is to find
a minimum number of vertex-disjoint simple directed paths, which contain all
terminal vertices and a minimum number of non-terminal vertices (Steiner
vertices).  The primary minimization criteria is the number of paths.
We show how to compute in linear time a minimum Steiner path cover for
directed co-graphs. This leads to a linear time computation of an
optimal directed Steiner path on directed co-graphs, if it exists.  
Since the Steiner path problem generalizes the Hamiltonian path problem,
our results imply the first linear time algorithm for the 
directed Hamiltonian path problem on directed co-graphs.
We also give binary integer programs for the (directed) Hamiltonian
path problem, for the (directed) Steiner path problem, and for the
(directed) Steiner path cover problem. These integer programs can be
used to minimize change-over times in pick-and-place machines used
by companies in electronic industry.

\end{abstract}

%%%%%%%%%%%%%%%%%%%%%%%%%%%%%%%%%%%%%%%%%%%%%%%%%%%%%%%%%%%%%%%%%%%%%%%%%%%
%%%%%%%%%%%%%%%%%%%%%%%%%%%%%%%%%%%%%%%%%%%%%%%%%%%%%%%%%%%%%%%%%%%%%%%%%%%
%%%%%%%%%%%%%%%%%%%%%%%%%%%%%%%%%%%%%%%%%%%%%%%%%%%%%%%%%%%%%%%%%%%%%%%%%%%

%%%%%%%%%%%%%%%%%%%%%%%%%%%%%%%%%%%%%%%%%%%%%%%%%%%%%%%%%%%%%%%%%%%%%%%%%%
\section{Introduction}
%%%%%%%%%%%%%%%%%%%%%%%%%%%%%%%%%%%%%%%%%%%%%%%%%%%%%%%%%%%%%%%%%%%%%%%%%%

For the well known Steiner tree problem there are efficient algorithms
on special graph classes like series-parallel graphs \cite{WC83},
outerplanar graphs \cite{WC82} and graphs of bounded
tree-width \cite{BCKN15,CMZ12}.
The class Steiner tree problem (CSP) is a generalization of the 
Steiner tree problem in which the vertices are partitioned into
classes of terminals \cite{RW90}. The unit-weight version of CSP can be
solved in linear time on co-graphs \cite{WY95}.

The Steiner path problem is a restriction of the Steiner
tree problem such that the required terminal vertices lie
on a path of minimum cost.
It is also a generalization of the Hamiltonian path problem:
If we choose each vertex as a terminal vertex, the Steiner path
problem becomes the Hamiltonian path problem.
The Euclidean bottleneck Steiner path problem
was considered in \cite{ACKS14} and a linear time solution for the
Steiner path problem on trees was given in \cite{MJV13}. 

While a Steiner tree always exists within connected graphs,
it is not always possible to find a Steiner path, which motivates
us to consider Steiner path cover problems. 
The Steiner path cover problem on interval graphs was considered
in \cite{CL18}.

In this article we consider the directed Steiner path cover problem
defined as follows. Let $G$ be a directed graph on vertex set $V(G)$ and
edge set $E(G)$ and let $T \subseteq V(G)$ be a set of terminal vertices.
Let $c: E(G) \to \RR^{\geq 0}$ be a function that assigns a weight to
each edge.
A {\em directed Steiner path cover} for $G$ is a set $P$ of
vertex-disjoint simple directed paths in $G$ that contain all
terminal vertices of $T$ and possibly also some of the non-terminal
(Steiner) vertices of $V(G) - T$.
The {\em size} of a directed Steiner path cover is the number of its paths,
i.e.~the size is $|P|$,
the {\em cost} is defined as the sum of weights of those edges used in
the paths in a directed Steiner path cover of minimum size.

\begin{desctight}
\item[Name] Directed Steiner Path Cover

\item[Instance] A directed graph $G$,
a set of terminal vertices $T \subseteq V(G)$,
and edge weights $c:  E(G) \to \RR^{\geq 0}$.

\item[Task] Find a directed Steiner path cover $P$ of minimum size 
for $G$ that minimizes $\sum_{p\in P}\sum_{e\in p} c(e)$.
\end{desctight}

To minimize only the sum of the edge weights makes no sense, because this
sum becomes minimal, if each terminal vertex is chosen as path of length 0.
So we primary have to demand that the number of paths should be minimal.

\smallskip
The directed Steiner path cover problem is NP-hard since it
generalizes the directed Hamiltonian path problem.
This motivates us to restrict the problem to special inputs. 
We consider a very natural class of inputs, which is defined as
follows.
Directed co-graphs (short for complement reducible graphs)  
can be generated from the single vertex graph by applying
disjoint union, order composition and series composition \cite{BGR97}.
They also can be characterized by excluding eight forbidden induced 
sub-digraphs, see \cite[Figure 2]{CP06}.

Directed co-graphs are exactly the digraphs of
directed NLC-width 1 and a proper subset of the digraphs
of directed clique-width at most 2 \cite{BG18,GWY16}.
Directed co-graphs are also interesting from an algorithmic point of view 
since several hard graph problems can be solved in 
polynomial time  by dynamic programming along the tree structure of
the input graph, see \cite{BM14,Gur17a,GR18c,GKR19f,GKR19d,GKR20c,GKR21,GKR21a}. 
Moreover, directed co-graphs are very useful for the reconstruction
of the evolutionary history of genes or species using genomic
sequence data \cite{HSW17,NEMWH18}.

For graphs where all edges have the same weight, the above definition
of the directed Steiner path cover problem results in a minimum
number of Steiner vertices. Graphs without edge weights can be considered
as a special case of graphs with unit-edge weights. Since edge weights do
not occur in co-graphs, we use the following problem definition.

\begin{desctight}
\item[Name] Unit-Edge-Weight Directed Steiner Path Cover

\item[Instance] A directed graph $G$ and
a set of terminal vertices $T \subseteq V(G)$.

\item[Task] Find a directed Steiner path cover of minimum size for $G$
such that the number of Steiner vertices is minimal.
\end{desctight}

In this paper we show how the value of a directed Steiner path cover of
minimum size and cost for the disjoint union, order composition and series
composition of two digraphs can be computed in linear time from the
corresponding values of the involved digraphs.  Therefore,
we define a useful  normal form for directed Steiner path covers in 
digraphs which are defined by the order composition or series 
composition of two digraphs. 
Further we give an algorithm which constructs a directed Steiner path cover
of minimum  size and  cost  for a directed co-graph in linear time.

The (unit-weight) directed Steiner tree problem  
%\cite{HRW92}
is $\MSOA$-definable by \cite[Proposition 4.14]{GHKLOR14}. 
But this does not hold for the directed Steiner path (cover) problem, since
it is a generalization of the directed Hamiltonian path problem which 
is not $\MSOA$-definable.  Thus it is not possible to obtain our solutions
for the directed Steiner path (cover) problem using the fact that directed
co-graphs have bounded directed clique-width \cite{GWY16} and
the result from \cite{GHKLOR14}
stating that all $\MSOA$-definable  digraph problems
are fixed parameter tractable  for the parameter directed clique-width.
Since the Hamiltonian cycle problem 
is $\w[1]$-hard parameterized by clique-width \cite{FGLS10a}
it even follows that Hamiltonian path problem and thus the
(directed) Steiner path (cover) problem are $\w[1]$-hard when parameterized
by (directed) clique-width.

We also give binary integer programs for the directed Hamiltonian path
problem and for the directed Steiner path (cover) problem.
These integer programs can be used to minimize change-over times in
pick-and-place machines used in electronic industry. The problem of
minimizing change-over times is introduced in section \ref{sec_hp}.

%\begin{desctight}
%\item[Name] Directed Hamiltonian Path
%
%\item[Instance] A directed graph $G$.
%
%\item[Task] Find a directed Hamiltonian path in $G$.
%\end{desctight}

\begin{desctight}
\item[Name] Directed Steiner Path

\item[Instance] A directed graph $G$, a set $T \subset V(G)$
  of terminal vertices, and a function $c: E(G) \to \RR^{\geq 0}$ that
  assigns each edge some weight.

\item[Task] Find a directed Steiner path $p$ in graph $G$ that minimizes
  $\sum_{e \in p} c(e)$.
\end{desctight}

%%%%%%%%%%%%%%%%%%%%%%%%%%%%%%%%%%%%%%%%%%%%%%%%%%%%%%%%%%%%%%%%%%%%%%%%%%
%%%%%%%%%%%%%%%%%%%%%%%%%%%%%%%%%%%%%%%%%%%%%%%%%%%%%%%%%%%%%%%%%%%%%%%%%%
\section{Preliminaries} \label{secsolDiCoGraph}
%%%%%%%%%%%%%%%%%%%%%%%%%%%%%%%%%%%%%%%%%%%%%%%%%%%%%%%%%%%%%%%%%%%%%%%%%%
%%%%%%%%%%%%%%%%%%%%%%%%%%%%%%%%%%%%%%%%%%%%%%%%%%%%%%%%%%%%%%%%%%%%%%%%%%

We consider the directed Steiner path cover problem on directed co-graphs,
so we first recall the definition of directed co-graphs and we will consider
a normal form of directed Steiner path covers.

%%%%%%%%%%%%%%%%%%%%%%%%%%%%%%%%%%%%%%%%%%%%%%%%%%%%%%%%%%%%%%%%%%%%%%%%%%
\subsection{Directed Co-Graphs}\label{secDiCoGraph}
%%%%%%%%%%%%%%%%%%%%%%%%%%%%%%%%%%%%%%%%%%%%%%%%%%%%%%%%%%%%%%%%%%%%%%%%%%

Directed co-graphs 
have been introduced  by Bechet et al.\ in   \cite{BGR97}.

\begin{definition}
The class of {\em directed co-graphs} is recursively defined as follows.
\begin{enumerate}[(i)]
\item Every digraph on a single vertex $(\{v\},\emptyset)$,
  denoted by $\bullet_v$, is a  {\em directed co-graph}.

\item If  $A$, $B$  are vertex-disjoint directed co-graphs, then
  \begin{enumerate}
  \item the disjoint union $A \oplus B$, which is defined as the digraph
    with vertex set $V(A) \cup V(B)$ and edge set $E(A) \cup E(B)$,

  \item the  order composition  $A \oslash B$, defined by their disjoint
    union plus all possible edges only directed from $V(A)$ to $V(B)$, and

  \item the series composition $A \otimes B$, defined by their disjoint
    union plus all possible edges between $V(A)$ and $V(B)$ in both
    directions,
  \end{enumerate}
  are {\em directed co-graphs}.
\end{enumerate}
\end{definition}

Every expression using these operations is called a {\em directed
co-expression}. The recursive generation of a directed co-graph can
be described by a tree structure, called {\em directed co-tree}.
The leaves of the directed co-tree represent the
vertices of the digraph and the inner vertices of the directed co-tree
correspond to the operations applied on the subgraphs of $G$ defined
by the subtrees. For every directed co-graph one can construct
a directed co-tree in linear time, see \cite{CP06}.

Directed co-graphs can also be characterized by excluding eight forbidden
induced sub-digraphs, see \cite[Figure 2]{CP06} or Figure~\ref{F-co-ex}.

\begin{table}[hbtp]
\begin{center}
\begin{tabular}{cccc}
\epsfig{figure=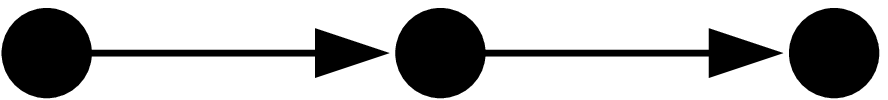,width=1.9cm} &
  \epsfig{figure=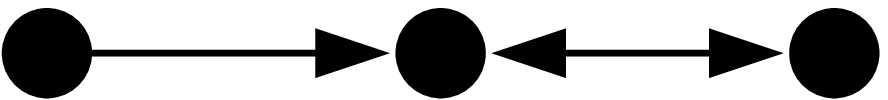,width=1.9cm} &
  \epsfig{figure=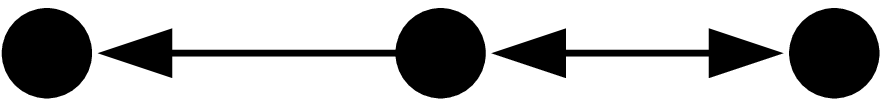,width=1.9cm} &
  \epsfig{figure=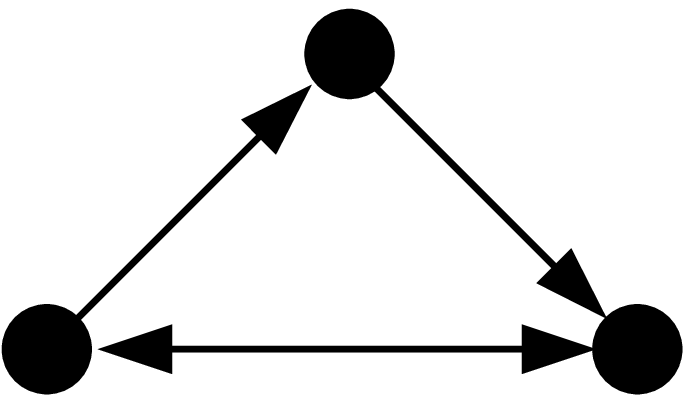,width=1.5cm} \\
 $D_1$    &    $D_2$    &   $D_3$   &  $D_4$  \\[1ex]
\epsfig{figure=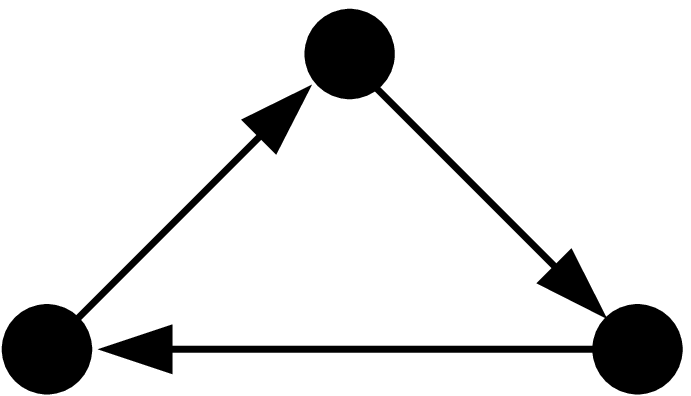,width=1.5cm} &
  \epsfig{figure=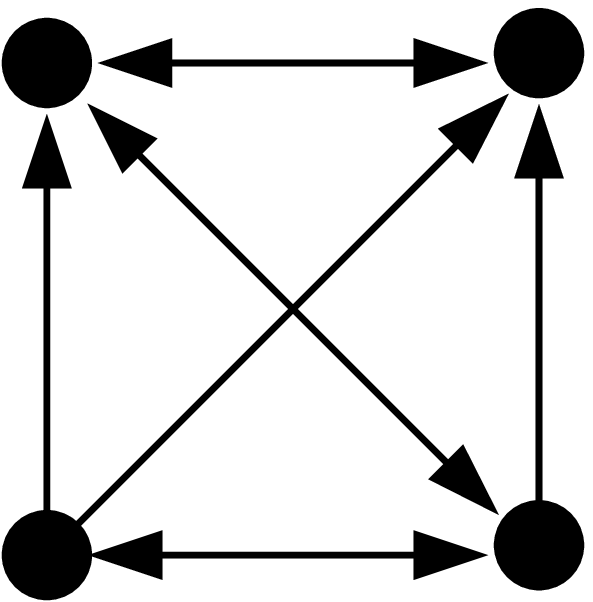,width=1.3cm} &
  \epsfig{figure=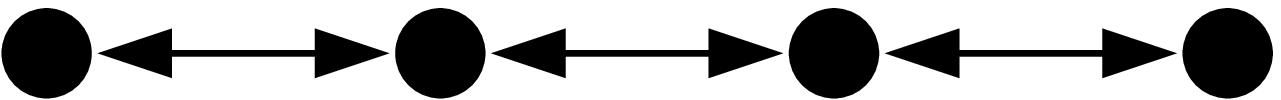,width=2.7cm} &
  \epsfig{figure=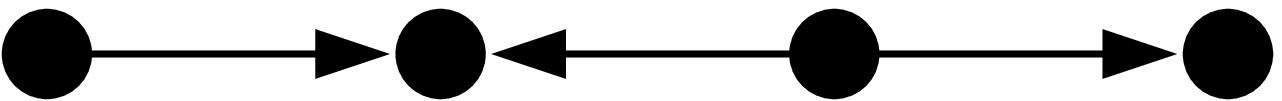,width=2.7cm} \\
 $D_5$   & $D_6$   &    $D_7$   &  $D_8$ \\
\end{tabular}
\end{center}
\caption{The eight forbidden induced sub-digraphs for directed co-graphs
  (see \cite{CP06}).}
\label{F-co-ex}
\end{table}

Next we define a normal form for directed Steiner path covers in 
digraphs which are defined by the order composition or series 
composition of two digraphs.

%%%%%%%%%%%%%%%%%%%%%%%%%%%%%%%%%%%%%%%%%%%%%%%%%%%%%%%%%%%%%%%%%%%%%%%%%%
\subsection{Normal Form for  Directed Steiner Path Covers}
%%%%%%%%%%%%%%%%%%%%%%%%%%%%%%%%%%%%%%%%%%%%%%%%%%%%%%%%%%%%%%%%%%%%%%%%%%

Let $G$ be a directed co-graph, let $T \subseteq V(G)$ be a set of terminal
vertices, and let $C$ be a directed Steiner path cover for $G$ with respect
to $T$. Then $\s(C)$ denotes the number of Steiner vertices in the paths
of $C$.

\begin{lemma}
\label{L_weakNF}
Let $C$ be a directed Steiner path cover for some directed co-graph
$G = A \oslash B$ or $G = A \otimes B$ with respect to a set
$T \subseteq V(G)$ of terminal vertices. Then there is a directed Steiner
path cover $C'$ with respect to $T$ which does not contain paths $p$ and
$p'$ satisfying one of the structures (\ref{L1_1})-(\ref{L1_4}), such
that $|C| \geq |C'|$ and $\s(C)  \geq \s(C')$ applies.
Let $q_1,\ldots,q_4$ denote sub-paths which may be empty.
\begin{enumerate}
\item \label{L1_1} % 1
  $p=(x, q_1)$ or $p=(q_1,x)$  where $x \not\in T$. Comment: No path
  starts or ends with a Steiner vertex.
\item \label{L1_2} % 2
  $p=(q_1, u, x, v, q_2)$ where $u \in V(A)$, $v \in V(B)$, and
  $x \not\in T$. Comment: On a path, the neighbors $u,v$ of a Steiner
  vertex $x$ are both contained in the same digraph.
\item \label{L1_3} % 3
  $p=(q_1,x)$, $p'=(u, q_2)$, where $x \in V(A)$, $u \in V(B)$,
  $p \neq p'$. Comment: No path $p$ ends in $A$, if there
  is a path $p' \neq p$ that starts in $B$.
\item \label{L1_4} % 4
  $p=(\ldots, x, u, v, y, \ldots)$ where $u,v \not\in T$.
  Comment: The paths contain no edge between two Steiner vertices.
\end{enumerate}
If $G = A \otimes B$ then cover $C'$ also does not contain paths
satisfying structures (\ref{L1_5})-(\ref{L1_8}).
\begin{enumerate}
\setcounter{enumi}{4}
\item \label{L1_5} % 5
  $p=(x, q_1)$, $p'=(u,q_2)$, where $x \in V(A)$, $u \in V(B)$,
  $p \neq p'$. Comment: All paths start in the same digraph.
\item \label{L1_6} % 6
  $p=(q_1, x, y, q_2)$, $p'=(q_3, u, v, q_4)$ where $x,y \in V(A)$,
  $u, v \in V(B)$. Comment: The cover $C'$ contains edges of only one
  of the digraphs.
\item \label{L1_7} % 7
  $p=(x, q_1)$, $p'=(q_2, u, y, v, q_3)$, where $x,y \in V(A)$,
  $u, v \in V(B)$, and $y \not\in T$. Comment: If a path starts in $A$
  then there is no Steiner vertex in $A$ with two neighbors on the path
  in $B$.
\item \label{L1_8} % 8
  $p=(x, q_1)$, $p'=(q_2, u, v, q_3)$, where $x \in V(A)$ and
  $u, v \in V(B)$. Comment: If a path starts in $A$, then no edge of $B$
  is contained in the cover.
\end{enumerate}
\end{lemma}

\begin{proof}
\begin{enumerate}
\item \label{P1_1} % 1
  If $x$ is removed from $p$ we get a cover with one Steiner vertex
  less than $C$.

  %\centerline{\epsfxsize=.8\textwidth \epsfbox{img/nf1.eps}}

\item \label{P1_2} % 2
  If $x$ is removed from $p$, we get a cover with one Steiner vertex
  less than $C$.

  %\centerline{\epsfxsize=.8\textwidth \epsfbox{img/nf2.eps}}

\item \label{P1_3} % 3
  We combine the paths to only one path $(q_1, x, u, q_2)$ and we get
  a cover with one path less than $C$.

  %\centerline{\epsfxsize=.6\textwidth \epsfbox{img/nf3.eps}}

\item \label{P1_4} % 4
  Since $G$ is a directed co-graph, the underlying undirected graph is
  a co-graph such that the path cannot include a $P_4$, i.e.\ a simple
  path of 4 vertices, as induced subgraph.
  Thus, there must be at least one additional arc. If such an
  additional arc would shorten the Steiner path by skipping $u$, $v$ or
  both then we remove $u$ or $v$ or both and take the shortcut for
  getting a cover $C'$. Additional arcs that do not shorten the path
  would create a forbidden induced subgraph from Figure~\ref{F-co-ex}
  which is not possible. For details see Table \ref{T_P4}.

\item \label{P1_5} % 5
  The new paths are $q_1$ and $(x,u.q_2)$. The cover $C'$ is as good as $C$.

  %\centerline{\epsfxsize=.95\textwidth \epsfbox{img/nf8.eps}}

\item \label{P1_6} % 6
  If $p \neq p'$, then $(q_1, x, v, q_4)$ and $(q_3, u, y, q_2)$ are the
  paths in cover $C'$.

  \centerline{\epsfxsize=.7\textwidth \epsfbox{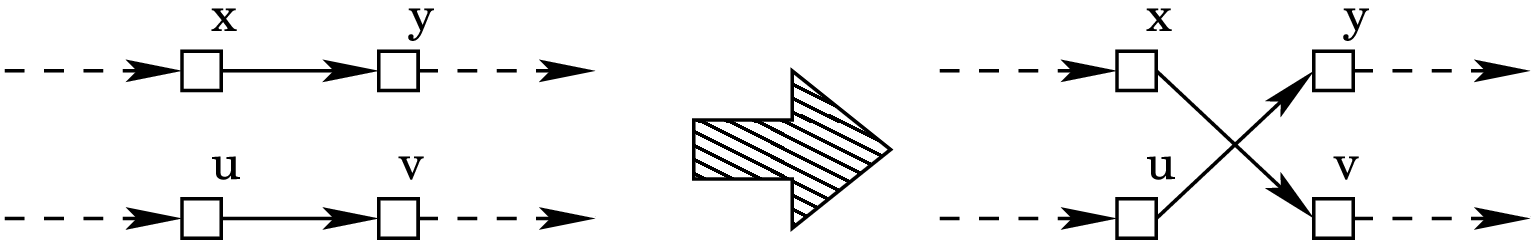}}

  If $p=p'$, then we have to distinguish whether
  $(u,v)\in q_1$,  $(u,v)\in q_2$, $(x,y)\in q_3$, or  $(x,y)\in q_4$. We
  show how to handle the first case, the other three cases are similar.
  Let $p=(q_3,u,v,q_5,b,a,q_6,x,y,q_2)$, where $b\in V(B)$ and $a\in V(A)$.
  Then the new path in cover $C'$ is $(q_3,u,a,q_6,x,v,q_5,b,y,q_2)$. Such
  vertices $a$ and $b$ must exist because $v \in V(B)$ and $x \in V(A)$,
  possibly it holds $a=x$ or $b=v$. In any case cover $C'$ is as good as $C$.

  \centerline{\epsfxsize=.7\textwidth \epsfbox{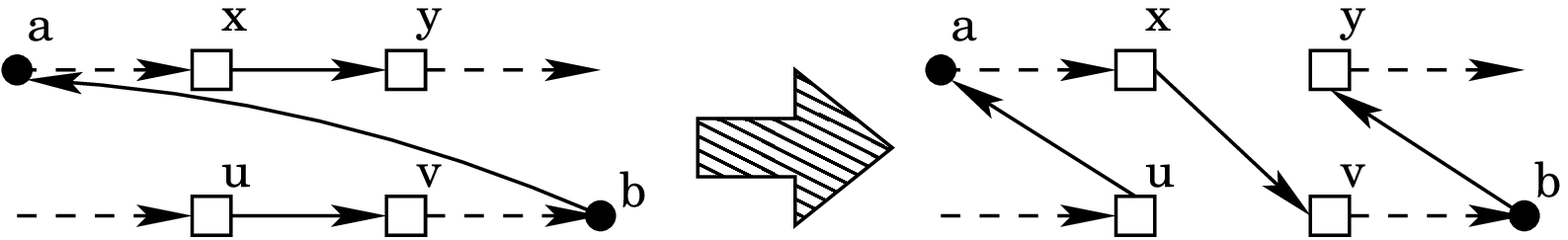}}

\item \label{P1_7} % 7
  If $p \neq p'$, then $q_1$ and $(q_2,u,x,v,q_3)$ are the new paths
  in cover $C'$. If $p=p'$, i.e.\ $q_1 = (q'_2,u,y,v,q_3)$, where $q'_2$ 
  is obtained from $q_2$ by removing $x$, then
  $(q'_2,u,x,v,q_3)$ is the new path in cover $C'$. The cover $C'$
  is as good as $C$. If $p \neq p'$, then the edge $(a,b)$ is missing
  in the following figure.

  %{\epsfxsize=.8\textwidth \epsfbox{img/nf4a.eps}} \\
  %{\epsfxsize=.8\textwidth \epsfbox{img/nf4b.eps}}
  \centerline{\epsfxsize=.7\textwidth \epsfbox{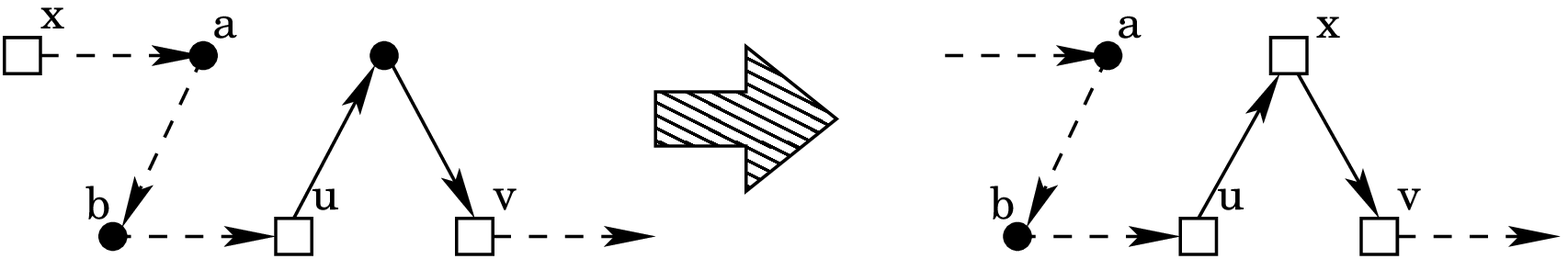}}

\item \label{P1_8} % 8
  If $p \neq p'$, then $q_1$ and $(q_2,u,x,v,q_3)$ are the new paths
  in cover $C'$. If $p=p'$, i.e.\ $q_1 = (q'_2,u,v,q_3)$,  where $q'_2$ 
  is obtained from $q_2$ by removing $x$, then
  $(q'_2,u,x,v,q_3)$ is the new path in cover $C'$. The cover $C'$
  is as good as $C$. If $p \neq p'$, then the edge $(a,b)$ is missing 
  in the following figure.

  %\centerline{\epsfxsize=.8\textwidth \epsfbox{img/nf5.eps}}
  \centerline{\epsfxsize=.7\textwidth \epsfbox{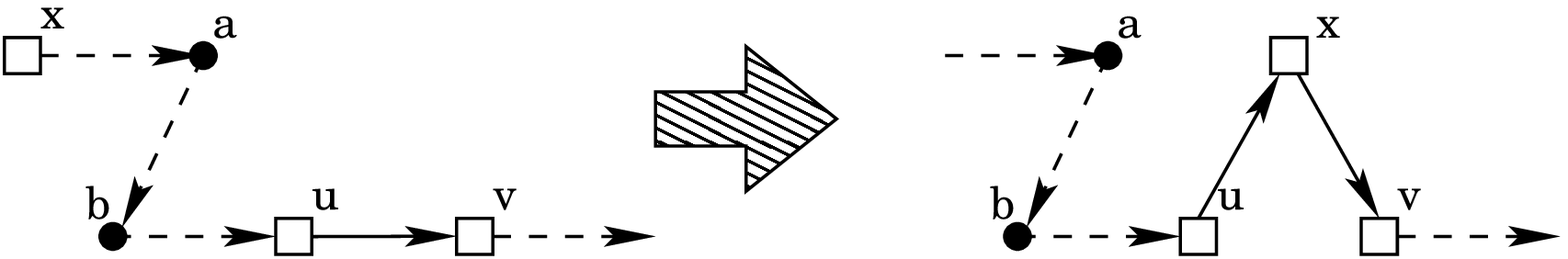}}

\end{enumerate}
Operations 1, 2, 4 and 7 reduce the number of Steiner vertices by one,
the remaining operations 3, 5 and 6 do not change the number of Steiner
vertices. Therefore, operations 1, 2, 4 and 7 can only be executed at
most $|V - (T_A \cup T_B)|$ times.

Operation 6 reduces the number of paths by one, the remaining operations
do not increase the number of paths. Therefore operation 6 can be executed
at most $\max\{|T_A|, |T_B|\}$ times.

Let us now consider those edges on a path that connect vertices of $A$ and
vertices of $B$. The maximum number of those edges is $|V(A)| + |V(B)| - 1$.
Operation 7 can remove two such edges, operations 3 and 5 can add two
such edges. Since the other operations 1, 2, 4 and 6 do not reduce the
number of edges, operations 3 and 5 can be used at most
$(|V(A)| + |V(B)| - 1) / 2 + |V - (T_A \cup T_B)|$ times.
\end{proof}

Since the hypothesis of Lemma \ref{L_weakNF} is symmetric in $A$
and $B$, the statement of Lemma \ref{L_weakNF} is also valid for
co-graphs $G = A \otimes B$ if $A$ and $B$ are switched.

\begin{table}[hbtp]
\begin{center}
\begin{tabular}{|c||c|c|c|c|}
\hline
graph with underlying $P_4$ and additional & none & $a$ & $b$ & $c$ \\
\cline{2-5}
edges that do not shorten the path &
      $a$, $b$ & $a$, $c$ & $b$, $c$ & $a$, $b$, $c$ \\
\hline
\hline
%\multirow{4}{*}{\epsfig{figure=img/c0.eps,width=3.5cm}} &
% \multirow{2}{*}{$D_1$} & \multirow{2}{*}{$D_1$} &
% \multirow{2}{*}{$D_2$} & \multirow{2}{*}{$D_1$} &
% \multirow{2}{*}{$D_3$} & \multirow{2}{*}{$D_2$} &
% \multirow{2}{*}{$D_2$} & \multirow{2}{*}{$D_7$} \\
% & & & & & & & & \\
% & $\{1,2,3\}$ & $\{2,3,4\}$ & $\{1,2,3\}$ & $\{1,2,3\}$
% & $\{2,3,4\}$ & $\{2,3,4\}$ & $\{2,3,4\}$ & $\{1,2,3,4\}$ \\
% & & & & & & & & \\
%\hline
\multirow{4}{*}{\epsfig{figure=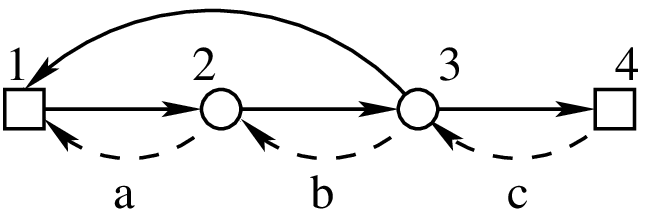,width=3.1cm}} &
 $D_5$ & $D_4$ & $D_4$ & $D_3$ \\
 & $\{1,2,3\}$ & $\{1,2,3\}$ & $\{1,2,3\}$ & $\{1,3,4\}$ \\
 \cline{2-5}
 & $D_3$ & $D_4$ & $D_4$ & $D_3$ \\
 & $\{2,3,4\}$ & $\{1,2,3\}$ & $\{1,2,3\}$ & $\{1,3,4\}$ \\
\hline \hline
\multirow{4}{*}{\epsfig{figure=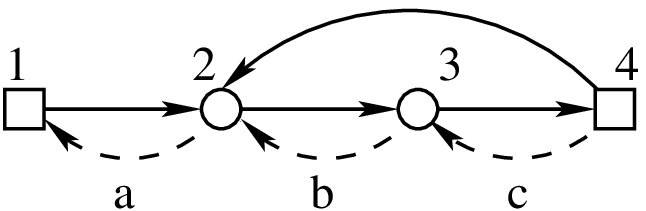,width=3.1cm}} &
 $D_5$ & $D_2$ & $D_2$ & $D_4$ \\
 & $\{2,3,4\}$ & $\{1,2,4\}$ & $\{1,2,3\}$ & $\{2,3,4\}$ \\
 \cline{2-5}
 & $D_2$ & $D_4$ & $D_2$ & $D_2$ \\
 & $\{1,2,4\}$ & $\{2,3,4\}$ & $\{1,2,3\}$ & $\{1,2,4\}$ \\
\hline \hline
\multirow{4}{*}{\epsfig{figure=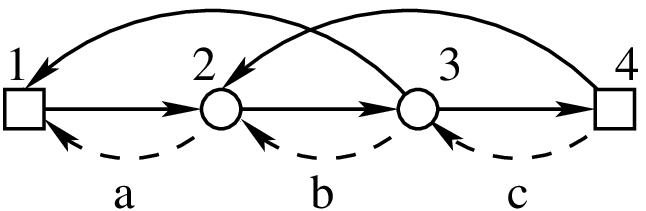,width=3.1cm}} &
 $D_5$ & $D_4$ & $D_4$ & $D_4$ \\
 & $\{2,3,4\}$ & $\{1,2,3\}$ & $\{1,2,3\}$ & $\{2,3,4\}$ \\
 \cline{2-5}
 & $D_4$ & $D_4$ & $D_4$ & $D_2$ \\
 & $\{2,3,4\}$ & $\{2,3,4\}$ & $\{1,2,3\}$ & $\{1,2,4\}$ \\
\hline \hline
\multirow{4}{*}{\epsfig{figure=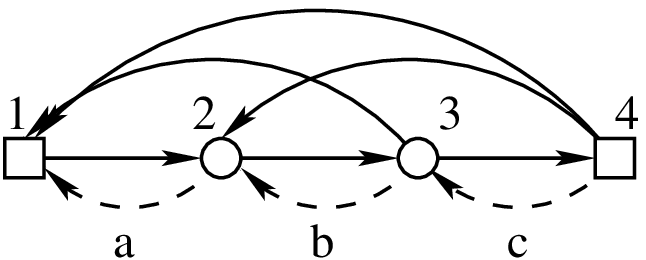,width=3.1cm}} &
 $D_5$ & $D_4$ & $D_4$ & $D_4$ \\
 & $\{1,2,3\}$ & $\{1,2,3\}$ & $\{1,2,3\}$ & $\{2,3,4\}$ \\
 \cline{2-5}
 & $D_4$ & $D_4$ & $D_4$ & $D_6$ \\
 & $\{2,3,4\}$ & $\{2,3,4\}$ & $\{1,2,3\}$ & $\{1,2,3,4\}$ \\
\hline \hline
\multirow{4}{*}{\epsfig{figure=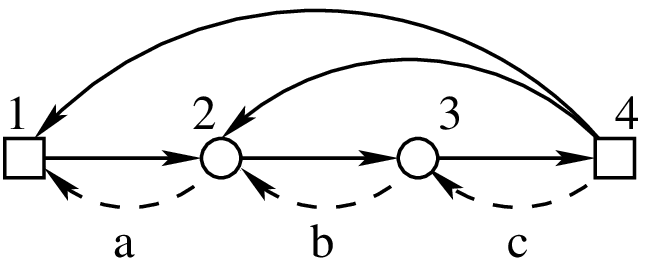,width=3.1cm}} &
 $D_5$ & $D_3$ & $D_2$ & $D_1$ \\
 & $\{2,3,4\}$ & $\{1,2,3\}$ & $\{1,2,3\}$ & $\{1,2,3\}$ \\
 \cline{2-5}
 & $D_1$ & $D_3$ & $D_3$ & $D_3$ \\
 & $\{1,3,4\}$ & $\{1,3,4\}$ & $\{1,3,4\}$ & $\{1,3,4\}$ \\
\hline \hline
\multirow{4}{*}{\epsfig{figure=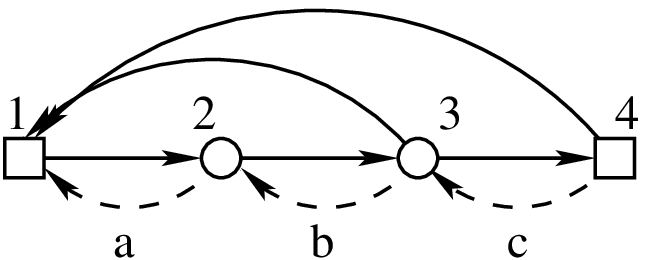,width=3.1cm}} &
 $D_5$ & $D_1$ & $D_1$ & $D_1$ \\
 & $\{1,2,3\}$ & $\{2,3,4\}$ & $\{1,2,4\}$ & $\{1,2,4\}$ \\
 \cline{2-5}
 & $D_2$ & $D_2$ & $D_1$ & $D_2$ \\
 & $\{1,2,4\}$ & $\{1,2,4\}$ & $\{1,2,4\}$ & $\{1,2,4\}$ \\
\hline \hline
\multirow{4}{*}{\epsfig{figure=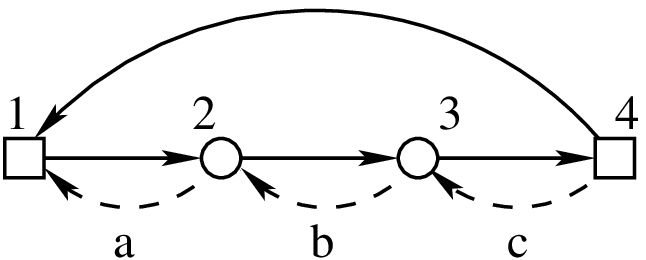,width=3.1cm}} &
 $D_1$ & $D_1$ & $D_1$ & $D_1$ \\
 & $\{1,2,3\}$ & $\{2,3,4\}$ & $\{1,2,4\}$ & $\{1,2,3\}$ \\
 \cline{2-5}
 & $D_3$ & $D_2$ & $D_2$ & $D_3$ \\
 & $\{2,3,4\}$ & $\{1,2,4\}$ & $\{1,2,3\}$ & $\{1,3,4\}$ \\
\hline
\end{tabular}
\end{center}
\caption{The leftmost column shows a graph with underlying undirected $P_4$
  and at least one additional arc that do not shorten the path. The other
  columns shows the forbidden subgraphs that are contained in the leftmost
  graph depending on the edges of the $P_4$.}
\label{T_P4}
\end{table}

\begin{definition}
A directed Steiner path cover $C$ for some directed co-graph
$G = A \oslash B$ or $G = A \otimes B$ is said to be in
{\em normal form} if none of the operations described in the proof
of Lemma \ref{L_weakNF} is applicable.
\end{definition}

In the following we assume that a directed Steiner path cover for some 
directed co-graph $G = A \oslash B$ or $G = A \otimes B$ is always in
normal form, since the operations of the proof of Lemma \ref{L_weakNF}
do not increase the number of paths or Steiner vertices of a cover.
Lemma \ref{L_weakNF} implies the following theorem.

\begin{theorem}\label{T_strongNF}
For each directed co-graph $G = A \otimes B$ and set of terminal vertices
$T \subseteq V(G)$ any directed Steiner path cover $C$ in normal form with
respect to $T$ does not contain an edge of digraph $A$,
and no path in $C$ starts or ends in digraph $A$
if $|T_A| < |T_B|$.
\end{theorem}

\begin{proof}[by contradiction]
Assume, the Steiner path cover $C$ contains an edge of digraph $A$.
Then by Lemma \ref{L_weakNF}(\ref{L2_5}), all paths starts in digraph $A$.
By Lemma \ref{L_weakNF}(\ref{L2_4}), it holds that no Steiner vertex $v$
of $V(A)$ is contained
in $C$, where the neighbors of $v$ are both of digraph $B$. By Lemma
\ref{L_weakNF} (\ref{L2_1}), (\ref{L2_2}), and (\ref{L2_5}), it holds
that all vertices of $V(B)$ from $C$ are connected with a terminal
vertex of $V(A)$, thus $|T_A| > |T_B|$.
\lightning

Second, we have to show that
no path in $C$ starts or ends in digraph $A$. Assume on the contrary, that
there is one path that starts in $A$. By Lemma \ref{L_weakNF}(\ref{L2_6}),
it holds that all paths start in $A$. Continuing as in the first case this
leads to a contradiction.
\end{proof}

\begin{remark}\label{R2_strongNF}
For each directed co-graph $G = A \oslash B$ and set of terminal vertices
$T \subseteq V(G)$ any directed Steiner path cover $C$ in normal form
with respect to $T$ it holds that each path that starts in $A$ 
either remains in $A$ or it crosses over to $B$ and remains in $B$.
Each path that reaches a vertex of $B$ has to stay in $B$ since no edge
from a vertex in $B$ to a vertex in $A$ exists.
\end{remark}

%%%%%%%%%%%%%%%%%%%%%%%%%%%%%%%%%%%%%%%%%%%%%%%%%%%%%%%%%%%%%%%%%%%%%%%%%%
\section{Algorithms for the Directed Steiner Path Cover Problem}
%%%%%%%%%%%%%%%%%%%%%%%%%%%%%%%%%%%%%%%%%%%%%%%%%%%%%%%%%%%%%%%%%%%%%%%%%%

%%%%%%%%%%%%%%%%%%%%%%%%%%%%%%%%%%%%%%%%%%%%%%%%%%%%%%%%%%%%%%%%%%%%%%%%%%
\subsection{Computing the optimal number of paths}
%%%%%%%%%%%%%%%%%%%%%%%%%%%%%%%%%%%%%%%%%%%%%%%%%%%%%%%%%%%%%%%%%%%%%%%%%%

Let $G$ be a directed co-graph and $T \subseteq V(G)$ be a set of terminal
vertices. We define $p(G,T)$ as the minimum number of paths within a
Steiner path cover for $G$ with respect to $T$. Further let $s(G,T)$ be
the minimum number of Steiner vertices in a  directed Steiner path cover
of size $p(G,T)$ with respect to $T$. We do not specify set $T$ if it is
clear from the context which set is meant.

\begin{lemma}\label{le-d-p}
Let $A$ and $B$ be two vertex-disjoint digraphs and let
$T_A \subseteq V(A)$ and $T_B \subseteq V(B)$ be two sets of terminal
vertices.  Then the following equations hold true:
\begin{enumerate}
\item \label{L2_1}
  $p(\bullet_v, \emptyset)=0$ and $p(\bullet_v, \{v\})=1$

\item \label{L2_2}
  $p(A \oplus B, T_A \cup T_B) = p(A, T_A) + p(B, T_B)$

\item \label{L2_3}
  $p(A \otimes B, \emptyset) = 0$

\item \label{L2_4}
  $p(A \otimes B, T_A \cup T_B) = \max\{1, p(B, T_B) - |V(A)|\}$
  if $1 \leq |T_B|$ and  $|T_A| \leq |T_B|$

\item \label{L2_5}
  $p(A \otimes B, T_A \cup T_B) = \max\{1, p(A, T_A) - |V(B)|\}$
  if $1 \leq |T_A|$ and  $|T_A| > |T_B|$

\item \label{L2_6}
  $p(A \oslash B, T_A \cup T_B) = p(A, T_A)$ if $p(A) \geq p(B)$

\item \label{L2_7}
  $p(A \oslash B, T_A \cup T_B) = p(B, T_B)$ if $p(A) < p(B)$

\end{enumerate}
\end{lemma}

\begin{proof}
\ref{L2_1}. - \ref{L2_3}. Obvious.
\begin{enumerate}
\setcounter{enumi}{3}
\item \label{P2_5}
  We show that $p(A \otimes B) \geq \max\{1,p(B)-|V(A)|\}$
  applies by an indirect proof. Assume a directed Steiner path cover $C$ for
  $A \otimes B$ has less than $ \max\{1,p(B) - |V(A)|\}$
  paths. The removal of all vertices of $A$ from all paths in $C$
  gives a directed Steiner path cover of size $|C| + |V(A)| < p(B)$ for
  $B$. \lightning

  To see that $p(A \otimes B) \leq \max\{1, p(B) - |V(A)|\}$
  applies, consider that we can use any vertex of $A$ to combine two
  paths of the cover of $B$ to one path, since the series composition
  of $A$ and $B$ creates all directed edges between $A$ and $B$. If there
  are more terminal vertices in $T_A$ than there are paths in the cover
  of $B$, i.e.\ $p(B) < |T_A|$,
  then we have to split paths of $B$ and reconnect them by terminal
  vertices of $T_A$. This can always be done since $|T_A| \leq |T_B|$.

\item \label{P2_6}
  Similar to \ref{P2_5}.

\item \label{P2_7}
  To see that $p(A \oslash B) \leq p(A)$ applies, consider that
  we can append any path of $A$ by any path of $B$, see Lemma
  \ref{L_weakNF}(\ref{L2_3}). Since no edge
  between $B$ and $A$ is created, no path of $B$ can be extended
  by a path of $A$.

  We show that $p(A \oslash B) \geq p(A)$ applies by an indirect
  proof. Assume a directed Steiner path cover $C$ for $A \oslash B$
  contains less than $p(A)$ paths. The removal of all vertices of $B$
  from all paths in $C$ gives a Steiner path cover of size
  $|C| < p(A)$. \lightning

\item \label{P2_8}
  Similar to \ref{P2_7}.
\end{enumerate}
This shows the statements of the lemma.
\end{proof}

%%%%%%%%%%%%%%%%%%%%%%%%%%%%%%%%%%%%%%%%%%%%%%%%%%%%%%%%%%%%%%%%%%%%%%%%%%
\subsection{Computing the optimal number of Steiner vertices}
%%%%%%%%%%%%%%%%%%%%%%%%%%%%%%%%%%%%%%%%%%%%%%%%%%%%%%%%%%%%%%%%%%%%%%%%%%

%%%%%%%%%%%%%%%%%%%%%%%%%%%%%%%%%%%%%%%%%%%%%%%%%%%%%%%%%%%%%%%%%%%%%%
\begin{remark} \label{R_s}
For two vertex-disjoint directed co-graphs $A$, $B$ and two sets of
terminal vertices $T_A \subseteq V(A)$, $T_B \subseteq V(B)$ it holds
that $s(A \oplus B, T_A \cup T_B) = s(A, T_A) + s(B, T_B)$,
since the disjoint union does not create any new edges.
\end{remark}

%%%%%%%%%%%%%%%%%%%%%%%%%%%%%%%%%%%%%%%%%%%%%%%%%%%%%%%%%%%%%%%%%%%%%%
\begin{remark} \label{R_dir-s}
Let $G=A \oslash B$ be a directed co-graph, and let $C$ be a directed Steiner
path cover of $G$ such that $p=(q_1,u_1,x,q_2,v_1)$ is a path in $A$,
$p_1=(u_2, q_3)$ and $p_2=(v_2,q_4)$ are paths in $B$, all paths are
vertex-disjoint paths in $C$, where $x \not\in T$,
$u_1,u_2,v_1,v_2 \in T$, and $q_1,\ldots,q_4$ are sub-paths. Then we can
split $p$ at vertex $x$ into two paths, combine them with $p_1$ and $p_2$
to get $(q_1,u_1,u_2,q_3)$ and $(q_2,v_1,v_2,q_4)$ as new paths and we
get a Steiner path cover without increasing the cost. 
If $A$ and $B$ are switched we get $(u_2,q_3,q_1,u_1)$ and
$(v_2,q_4,q_2,v_1)$ as new paths and the statement also holds.
\end{remark}

What follows is the central lemma of our work, the proof is by induction
on the structure of the co-graph.

%%%%%%%%%%%%%%%%%%%%%%%%%%%%%%%%%%%%%%%%%%%%%%%%%%%%%%%%%%%%%%%%%%%%%%
\begin{lemma} \label{L_dir_sum}
For every directed co-graph $G$ and every  directed Steiner path cover
$C$ for $G$ with respect to a set $T \subseteq V(G)$ of terminal vertices
it holds that $p(G) + s(G) \leq |C| + \s(C)$.
\end{lemma}

\begin{proof}[by induction]
The statement is obviously valid for all directed co-graphs which consist
of only one vertex. Let us assume that the statement is valid for directed
co-graphs of $n$ vertices. Let $A$ and $B$ are vertex-disjoint directed
co-graphs of at most $n$ vertices each.

\smallskip \noindent {\bf Disjoint union:}
Let $G = A \oplus B$ be a directed co-graph that consists of more than
$n$ vertices.  By Lemma \ref{le-d-p}, and Remark \ref{R_s}, it holds that
$p(A \oplus B) +s(A \oplus B) = p(A) + p(B) + s(A) + s(B)$. By the
induction hypothesis, it
holds that $p(A) + s(A) \leq |C_{|A}| + \s(C_{|A})$ and
$p(B) + s(B) \leq |C_{|B}| + \s(C_{|B})$, where $C_{|A}$ denotes the cover
$C$ restricted to digraph $A$, i.e.\ the cover that results from $C$ when
all vertices of $B$ are removed. Then the statement of the lemma follows.
\[
    p(A \oplus B) + s(A \oplus B) \leq
       |C_{|A}| + \s(C_{|A}) + |C_{|B}| + \s(C_{|B}) =
       |C| + \s(C)
\]

\noindent {\bf Series composition:}
Let $G = A \otimes B$ be a directed co-graph that consists of more than $n$
vertices.  Without loss of generality, let $|T_A| \leq |T_B|$.

\begin{enumerate}
\item Let $X(A)$ denote the vertices of $A$ used in cover $C$, and let
  $D$ denote the cover for $B$ that we obtain by removing the vertices
  of $X(A)$ from cover $C$. By the induction hypothesis, it holds that
  $p(B) + s(B) \leq |D| + \s(D)$.

\item Let $nt(X(A))$ denote the number of non-terminal vertices of
  $X(A)$.  By Theorem \ref{T_strongNF} it holds that
  $\s(C) = \s(D) + nt(X(A))$ and
  $|C| = |D| - |T_A| - nt(X(A))$.  Thus, we get
  $|C| + \s(C) = |D| + \s(D) - |T_A|$.
\end{enumerate}

We put these two results together and obtain:
\[
  p(B) + s(B) - |T_A| \leq |D| + \s(D) - |T_A| = |C| + \s(C)
\]
To show the statement of the lemma, we first consider the case
$p(B) - 1 \leq |V(A)|$.
Then it holds that $p(A \otimes B) = 1$. If $|T_A| \geq p(B) - 1$, then
$d := |T_A| - (p(B)-1)$ many Steiner vertices from $B$, if available,
can be replaced by terminal vertices from $A$.
Otherwise if $|T_A| < p(B) - 1$, then $-d = (p(B) - 1) - |T_A|$ many
Steiner vertices from $A$ are used to combine the paths.
Thus, it holds that $s(A \otimes B) \leq \max\{0,s(B) - d\}$
since the number of Steiner vertices in an optimal cover is at most
the number of Steiner vertices in a certain cover. Thus, since
$p(A \otimes B) = 1$ we get for $s(B)\geq d$:
\begin{eqnarray*}
p(A \otimes B) + s(A \otimes B) & \leq &
     1 + s(B) - d=  1 + s(B) - (|T_A| - (p(B)-1)) \\
 & = & \cancel{1} + s(B) - |T_A| + p(B) - \cancel{1} ~ \leq ~ |C| + \s(C)
\end{eqnarray*}
If  $s(B)< d$ then all Steiner vertices of $B$ can be replaced by terminal
vertices of $A$ and since $|T_A|\leq |T_B|$ holds, some of the paths 
of $B$ can be reconnected by the remaining terminal vertices of $A$. Thus,
$p(A \otimes B) + s(A \otimes B)=1\leq |C| + \s(C)$ applies.

Consider now the case where $p(B) - 1 > |V(A)|$ holds, i.e.\ not all paths
in an optimal cover for $B$ can be combined by vertices of $A$.  By Lemma
\ref{le-d-p}, it holds that $p(A \otimes B) = \max\{1, p(B) - |V(A)|\}$.
Thus, for $p(A \otimes B) > 1$ we get:
\begin{eqnarray*}
p(A \otimes B) + s(A \otimes B)
 & \leq & p(B) - |V(A)| + s(B) + nt(A) \\
 & = & p(B) + s(B) - |T_A| ~ \leq ~ |C| + \s(C)
\end{eqnarray*}
The non-terminal vertices of $A$ must be used to combine paths
of the cover, thus the non-terminal vertices of $A$ become Steiner
vertices.

\smallskip \noindent {\bf Order composition:}
Let $G = A \oslash B$ be a directed co-graph that consists of more than
$n$ vertices.  By the induction hypothesis,
it holds that $p(A) + s(A) \leq |C_{|A}| + \s(C_{|A})$ and
$p(B) + s(B) \leq |C_{|B}| + \s(C_{|B})$.

Let us first consider the case $p(A) > p(B)$. By Lemma \ref{le-d-p} it
holds $p(A \oslash B) = p(A)$. We can append any path of $A$ by any path
of $B$, and by Remark \ref{R_dir-s} it holds that for every path that there
is more in $A$ than in $B$, a Steiner vertex of $B$ can be removed. And
since an optimal cover has at most as many Steiner vertices as a concrete
cover, it holds
$s(A \oslash B) \leq \s(C_{|A}) + \s(C_{|B}) - \min \{\s(C_{|B}),
|C_{|A}| - |C_{|B}|\}$. If we sum up both equations we get
\[
   p(A \oslash B) + s(A \oslash B) \leq
      p(A) + \s(C_{|A}) + \s(C_{|B}) -
          \min \{\s(C_{|B}), |C_{|A}| - |C_{|B}|\}
\]
If $\s(C_{|B}) \geq |C_{|A}| - |C_{|B}|$ applies, and since
$\s(C) = \s(C_{|A}) + \s(C_{|B})$ applies, we get
\[
   p(A \oslash B) + s(A \oslash B) \leq p(A) + \s(C) - |C_{|A}| + |C_{|B}|.
\]
The statement would be shown if $p(A) - |C_{|A}| + |C_{|B}| \leq |C|$
would apply. It holds $p(A) \leq |C_{|A}|$, since an optimal cover has at
most as many paths as a concrete cover, and it holds $|C_{|B}| \leq |C|$,
since $|C| = \max\{|C_{|A}|, |C_{|B}|\}$ by Remark \ref{R2_strongNF}.
We sum up these
equations and we get $p(A) + |C_{|B}| \leq |C_{|A}| + |C|$, which is
equivalent to $p(A) - |C_{|A}| + |C_{|B}| \leq |C|$, thus
$p(A \oslash B) + s(A \oslash B) \leq |C| + \s(C)$ has been shown.

If $\s(C_{|B}) < |C_{|A}| - |C_{|B}|$, then it holds
$p(A \oslash B) + s(A \oslash B) \leq p(A) + \s(C_{|A})$, and we have
to show that $p(A) + \s(C_{|A}) \leq |C| + \s(C)$ applies. It holds
$p(A) \leq |C_{|A}|$, since an optimal cover has at most as many
paths as a concrete cover, and it holds $|C_{|A}| \leq |C|$, since
$|C| = \max\{|C_{|A}|, |C_{|B}|\}$ by Remark \ref{R2_strongNF}.
Furthermore, it holds $\s(C_{|A}) \leq \s(C)$, since a part is only
as big as the whole.

The other case $p(A) \leq p(B)$ can be shown in a similar way.
\end{proof}

To see why Lemma \ref{L_dir_sum} is crucial for the rest of this work,
consider the directed graph $B$ of Figure~\ref{F_contra} that is not a directed
co-graph. Terminal vertices $T_A=\{f,g\}$ and $T_B=\{a,c,e,u,w,x\}$ are
shown as squares.  In the left part of the figure a Steiner path cover
$C_\ell = \{(a,b,c,d,e), (u,v,w,x,y)\}$ for graph $B$ is shown with
$|C_\ell| = 2$ and $\s(C_\ell) = 4$ which is optimal. In the right part
of the figure a Steiner path cover $C_r = \{(a,b,c,w,x,y), (e), (u)\}$
for $B$ is shown with $|C_r| = 3$ and $\s(C_r) = 2$.
\begin{figure}[hbtp]
{\epsfxsize=.4\textwidth \epsfbox{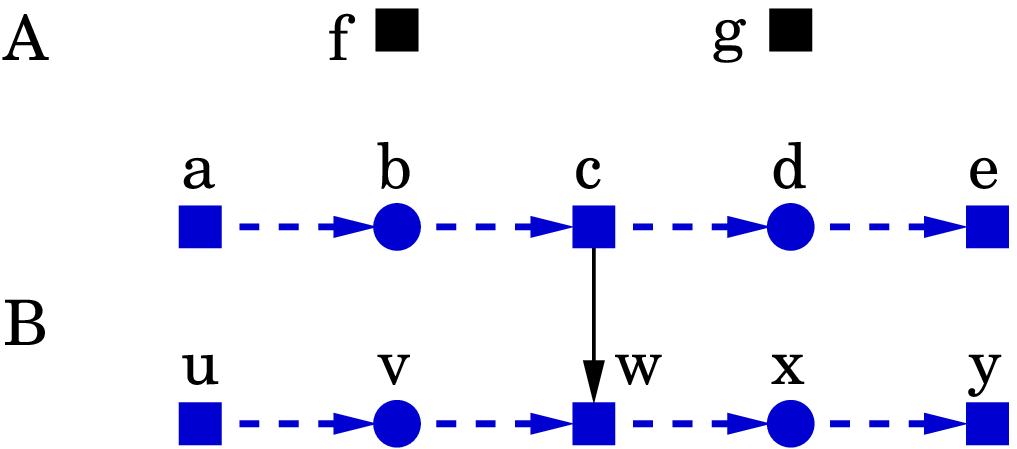}} \hfill
{\epsfxsize=.4\textwidth \epsfbox{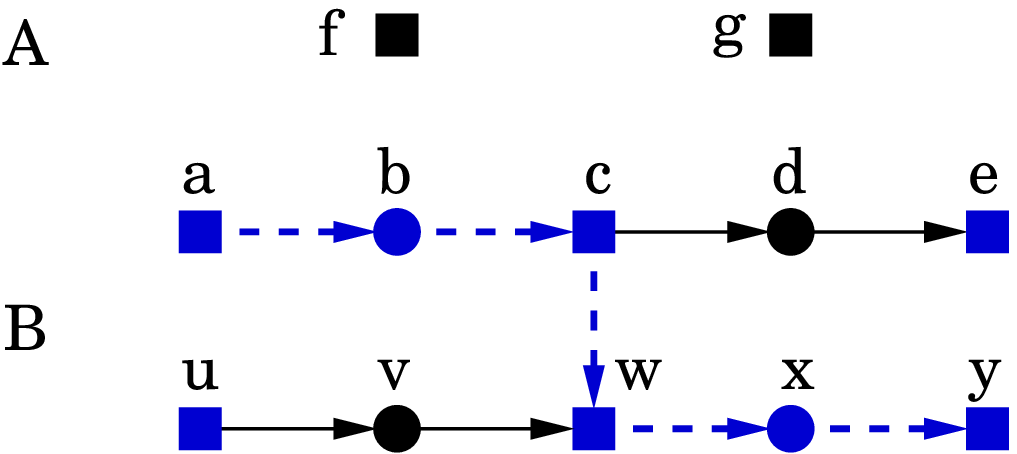}}
\caption{Small example that shows the contrapositive of the statement of
  Lemma \ref{L_dir_sum} in a graph $B$ that is not directed co-graph.}
\label{F_contra}
\end{figure}
The right cover can be extended to an optimal cover for $A \otimes B$ if
the vertices of $A$ are used to combine the path: $\{(u,f,a,b,c,w,x,y,g,e)\}$
is an optimal cover for $A \otimes B$ with only one path and 2 Steiner
vertices. The left Steiner path cover can not be extended to an optimal
cover for $A \otimes B$. For example, we can split path $(a,b,c,d,e)$ at
vertex $b$ into two paths $(a)$ and $(c,d,e)$ and reconnect them by a vertex
of $A$ and get $(a,f,c,d,e)$. The other vertex of $A$ must be
used to combine the remaining two paths to $(a,f,c,d,e,g,u,v,w,x,y)$ which
results in a cover for $A \otimes B$ that consists of one path but 3 Steiner
vertices. For graph $B$ the statement of Lemma \ref{L_dir_sum} is not
satisfied: $p(B) + s(B) = 2 + 4 = 6 > |C_r| + \s(C_r) = 3 + 2 = 5$

In the proof of Lemma \ref{le-d-s} we use the statement of Lemma
\ref{L_dir_sum} to show that optimal solutions for directed co-graphs
$A$ and $B$ can be combined to an optimal solution for $A \oslash B$
and $A \otimes B$.

%%%%%%%%%%%%%%%%%%%%%%%%%%%%%%%%%%%%%%%%%%%%%%%%%%%%%%%%%%%%%%%%%%%%%%
\begin{remark} \label{r-s}
Let $G$ be a directed co-graph and let $C$ be a directed Steiner path cover
for $G$ with respect to some set of terminal vertices $T \subseteq V(G)$.
Then $\s(C) \geq s(G)$ holds only if $|C| = p(G)$.
If $|C| > p(G)$ then $\s(C)$ might be smaller than $s(G)$.
\end{remark}

%This fact will be used in the proof of the next lemma.

\begin{lemma}\label{le-d-s}
Let $A$ and $B$ be two vertex-disjoint digraphs, and let
$T_A \subseteq V(A)$, $T_B \subseteq V(A)$ be sets of terminal
vertices. Then the following equations applies:
\begin{enumerate}
\item \label{L3_1}
  $s(\bullet_v, \emptyset)=0$ and $s(\bullet_v, \{v\})=0$

\item \label{L3_2}
  $s(A \oplus B, T_A \cup T_B) = s(A, T_A) + s(B, T_B)$

\item \label{L3_3}
  $s(A \otimes B) = \max\{0, s(B) + p(B) - p(A \otimes B) - |T_A|\}$
  if $|T_A| \leq |T_B|$
 
\item \label{L3_4}
  $s(A \otimes B) = \max\{0,s(A) + p(A) - p(A \otimes B) - |T_B|\}$
  if $|T_A| > |T_B|$

\item \label{L3_5}
  $s(A \oslash B) = s(A) + s(B)$ if $p(A) = p(B)$

\item \label{L3_6}
  $s(A \oslash B) = s(A) + s(B) - \min\{s(A), p(B)-p(A)\}$
  if $p(A) < p(B)$

\item \label{L3_7}
  $s(A \oslash B) = s(A) + s(B) - \min\{s(B), p(A)-p(B)\}$
  if $p(A) > p(B)$
\end{enumerate}
\end{lemma}

\begin{proof}
\begin{enumerate}
\item Obvious.
\item See Remark \ref{R_s}
\item First, we show $s(A \otimes B) \leq
\max\{0, s(B) + p(B) - p(A \otimes B) - |T_A|\}$.

By Lemma \ref{L_dir_sum}, we know that
$s(A \otimes B) + p(A \otimes B) \leq \s(C) + |C|$ holds
true for any cover $C$ for co-graph $A \otimes B$ and any set
of terminal vertices $T$. Consider cover $C$ for $A \otimes B$
obtained by an optimal cover $D$ for $B$ in the following way:
Use the terminal vertices of $A$ to either combine paths of $D$ or
to remove a Steiner vertex of $D$ by replacing $v \not\in T$ by some
terminal vertex of $A$ in a path like $(\ldots, u, v, w, \ldots) \in D$,
where $u,w \in T$. 
If $|T_A|\geq s(B)+p(B)$ then all paths of $D$ 
can be combined and all Steiner vertices can be removed by
terminal vertices of $A$ and since $|T_A|\leq |T_B|$
applies, some of the paths can be split and reconnected
by the remaining terminal vertices of $A$.
Thus, $\s(C) + |C|=1$ and $s(A \otimes B)=0$.

Otherwise, if $|T_A|< s(B)+p(B)$, then
we get $\s(C) + |C| = s(B) + p(B) - |T_A|$,
and by Lemma \ref{L_dir_sum}, we get the statement.
\[
\begin{array}{crcl}
  & s(A \otimes B) + p(A \otimes B) & \leq &
     s(B) + p(B) - |T_A|  ~ = ~ \s(C) + |C| \\
  \iff & s(A \otimes B) & \leq &
     s(B) + p(B) - p(A \otimes B) - |T_A|
\end{array}
\]

  We prove now that $s(A \otimes B) \geq
  \max\{0, s(B) + p(B) - p(A \otimes B) - |T_A|\}$.

  Let $X(A)$ be the vertices of $V(A)$ that are contained in the paths of
  an optimal cover $C$ for $A \otimes B$.
  Let $D$ be the cover for $B$ obtained by removing the vertices of
  $X(A)$ from $C$. Then by Theorem \ref{T_strongNF} %Lemma \ref{L_weakNF},
  the following applies:
\[
\begin{array}{crcl}
  & |X(A)| = nt(X(A)) + |T_A| & = & |D| - p(A \otimes B) \\
  \iff & nt(X(A)) & = & |D| - p(A \otimes B) - |T_A|
\end{array}
\]
  Thus, we get:
\[
\begin{array}{rcl}
  s(A \otimes B) - nt(X(A)) ~ = ~ \s(D)
          & = & s(A \otimes B) - |D| + p(A \otimes B) + |T_A| \\
\iff ~~~ s(A \otimes B) & = & \s(D) + |D| - p(A \otimes B) - |T_A| \\
\Rightarrow ~~~~~ 
    s(A \otimes B) & \geq & s(B) + p(B) - p(A \otimes B) - |T_A|
\end{array}
\]
The implication follows since by Lemma \ref{L_dir_sum} it holds
$\s(D) + |D| \geq s(B) +p(B)$.

\item Can be shown similar to the previous item.
\item To see that $s(A \oslash B) \leq s(A) + s(B)$ applies,
  consider optimal covers $C$ and $D$ for $A$ and $B$. We construct a
  cover $E$ for $A \oslash B$ in such a way that any path of 
  $C$ is
  appended by a path of $D$, see Lemma \ref{L_weakNF}(\ref{L2_3}).
  Since $|E| = p(A \oslash B)$ holds, we get
  $s(A \oslash B) \leq \s(E) = \s(C) + \s(D) = s(A) + s(B)$,
  because an optimal cover has at most as many Steiner vertices as
  a concrete cover.

  To see that $s(A \oslash B) \geq s(A) + s(B)$ applies consider
  an optimal cover $C$ for $A \oslash B$. Then it holds
  $s(A \oslash B) = \s(C_{|A}) + \s(C_{|B}) \geq s(A) + s(B)$, since
  $|C_{|A}| = p(A) = p( A \oslash B)= p(B) = |C_{|B}|$.

\item We have to distinguish two cases. First, let
  $s(A) > p(B) - p(A)$.

  To see that
  $s(A \oslash B) \leq s(A) + s(B) - (p(B) - p(A))$ applies,
  consider optimal covers $C$ and $D$ for $A$ and $B$. We construct
  a cover $E$ for $A \oslash B$ in such a way that we first split
  $p(B) - p(A)$ many paths of $C$ at Steiner vertices as described
  in Remark \ref{R_dir-s}. Afterwards we put together each of the resulting
  paths by a path of $D$. 
  Thus it holds $|E| = p(A \oslash B) = p(B)$
  and therefore $s(A \oslash B) \leq \s(C) + \s(D) - (p(B) - p(A)) =
  s(A) + s(B) - (p(B) - p(A))$.

  Please note, a Steiner path cover $C$ for $A \oslash B$ with
  $\s(C_{|A}) > 0$ is not optimal if $|C_{|A}| < |C| = p(A \oslash B)$
  holds. By Remark \ref{R_dir-s} a path of $C_{|A}$ could be splitted
  at a Steiner vertex and the number of Steiner vertices could be reduced.

  To see that $s(A \oslash B) \geq s(A) + s(B) - (p(B) - p(A))$
  applies, consider an optimal cover $C$ for
  $A \oslash B$. Then it holds $s(A \oslash B) = \s(C) = \s(C_{|A}) +
  \s(C_{|B})$, and by the previous note it holds
  $|C| = p(A \oslash B) = p(B) = |C_{|A}|$. By Lemma \ref{L_dir_sum}
  we get $\s(C_{|A}) + |C_{|A}| \geq s(A) + p(A)$. If we sum up these
  equations we get $s(A \oslash B) + p(A \oslash B) =
  \s(C_{|A}) + |C_{|A}| + \s(C_{|B})$. Finally we get:
\begin{eqnarray*}
  s(A \oslash B)
    & = & \s(C_{|A}) + |C_{|A}| - p(A \oslash B) + \s(C_{|B}) \\
    & \geq & s(A) + p(A) - p(B) + \s(C_{|B}) 
    ~ \geq ~ s(A) + p(A) - p(B) + s(B)
\end{eqnarray*}
  The last step holds since $p(B) = |C_{|B}|$ and by Remark \ref{r-s}.

  Consider now the case that $s(A) \leq p(B) - p(A)$. To see that
  $s(A \oslash B) \leq s(B)$ applies,
  consider optimal covers $C$ and $D$ for $A$ and $B$. We construct
  a cover $E$ for $A \oslash B$ in such a way that we first split
  as many paths of $C$ at Steiner vertices as possible in a way described
  in Remark \ref{R_dir-s}. Afterwards all Steiner vertices 
  of $C$ have been removed and
  we put together each of the resulting
  paths by a path of $D$. 
  Thus it holds $|E| = p(A \oslash B) = p(B)$
  and therefore $s(A \oslash B) \leq \s(E) = s(B)$.

  To see that $s(A \oslash B) \geq s(B)$ applies, consider an
  optimal cover $C$ for $A \oslash B$. By the above note it holds
  $\s(C_{|A}) = 0$, since $C$ would not be optimal otherwise. Thus, we
  get $s(A \oslash B) = \s(C_{|B}) \geq s(B)$, since
  $|C_{|B}| = p(B)$ holds and by Remark \ref{r-s}.

\item Can be shown similar to the previous item.
\end{enumerate}
This shows the statements of the lemma.
\end{proof}

By Lemma \ref{le-d-p} and \ref{le-d-s}, and since a directed co-tree can
be computed in linear time from the input directed co-graph \cite{CP06},
we have shown the following result.

\begin{theorem}
The value of a  directed Steiner path cover of minimum cost for a directed
co-graph can be computed in linear time with respect to the size of the
directed co-expression.
\end{theorem}

% min s+p
Lemma \ref{L_dir_sum} allows us to minimize the following additional
cost function.

\begin{corollary}
The value of a directed Steiner path cover $C$ for a directed co-graph
$G$ such that $|C| + \s(C)$ is minimal can be computed in linear time
with respect to the size of the directed co-expression.
\end{corollary}

%%%%%%%%%%%%%%%%%%%%%%%%%%%%%%%%%%%%%%%%%%%%%%%%%%%%%%%%%%%%%%%%%%%%%%%%%%
\subsection{Computing an Optimal Directed Steiner Path Cover}
%%%%%%%%%%%%%%%%%%%%%%%%%%%%%%%%%%%%%%%%%%%%%%%%%%%%%%%%%%%%%%%%%%%%%%%%%%

Now we want to give an algorithm to compute an optimal directed Steiner
path cover for some given directed co-graph. The function {\sc SeriesComp},
see Algorithm \ref{fig:algorithm3xa},
returns an optimal Steiner path cover for some co-graph $G = A \otimes B$,
if some optimal covers $C_A$ for $A$ and $C_B$ for $B$ are given as
parameter. Let $T_A \subseteq V(A)$ and $T_B \subseteq V(B)$ be the sets
of terminal vertices, and let $|T_A| \leq |T_B|$.
The function {\sc append}$(p,u)$ used in Algorithm \ref{fig:algorithm3xa} in
lines \ref{app1}, \ref{app2}, and \ref{app3} appends path $p$ by vertex $u$.
The function {\sc Combine}$(p, p', u)$ used in lines \ref{cat1} and
\ref{cat2}, combines path $p$ and $p'$ by vertex $u$.
The function {\sc Replace}$(p, v, u)$ used in line \ref{repl1} removes
the vertex $v$ from path $p$ and replaces it by vertex $u$.
The function {\sc Insert}$(p, v, u)$ used in line \ref{ins1} inserts
vertex $u$ between vertex $v$ and its successor in path $p$.

\begin{algorithm}[ht]
\footnotesize
\begin{algorithmic}[1]
%\Function{SeriesComp}{set $A$, Cover $C_A$, set $B$, Cover $C_B$}
\If {$|T_A| = |T_B|$ or $|T_A| = |T_B| - 1$}
      \Comment{results in only one path}
  \State let $T'_A := T_A$; let $T'_B := T_B$; let $p := ()$;
  \While {$T'_A \neq \emptyset$}
    \State let $b \in T'_B$; $T'_B := T'_B - \{b\}$;
         \Call{Append}{$p$, $b$}; \label{app1}
    \State let $a \in T'_A$; $T'_A := T'_A - \{a\}$;
         \Call{Append}{$p$, $a$}; \label{app2}
  \EndWhile
  \If {$T'_B \neq \emptyset$}
    \State let $b \in T'_B$; $T'_B := T'_B - \{b\}$;
         \Call{Append}{$p$, $b$} \label{app3}
  \EndIf
  \State \textbf{return} $\{ p \}$
\EndIf
\State let $T'_A := T_A$; let $C'_B := C_B$;
\State let $p \in C'_B$; $C'_B := C'_B - \{p\}$;

\While {$T'_A \neq \emptyset$ and $C'_B \neq \emptyset$}
      \Comment{combine 2 paths by a terminal vertex of $A$}
   \State let $a \in T'_A$; $T'_A := T'_A - \{a\}$;
   \State let $p' \in C'_B$; $C'_B := C'_B - \{p'\}$;
   \State \Call{Combine}{$p$, $p'$, $a$}; \label{cat1}
\EndWhile
\If {$T'_A = \emptyset$}
  \State let $U := V(A) - T_A$;
      \Comment{$U$ is the set of non-terminal vertices of $A$}
  \While {$U \neq \emptyset$ and $C'_B \neq \emptyset$}
      \Comment{combine 2 paths by a non-terminal vertex of $A$}
    \State let $u \in U$; $U := U - \{u\}$;
    \State let $p' \in C'_B$; $C'_B := C'_B - \{p'\}$;
    \State \Call{Combine}{$p$, $p'$, $u$} \label{cat2}
  \EndWhile
  \State \textbf{return} $\{ p \} \cup C'_B$
\EndIf
\If {$C'_B = \emptyset$}
  \While {$T_A \neq \emptyset$}
    \State let $a \in T'_A$; $T'_A := T'_A - \{a\}$;
    \If {$\exists$ Steiner vertex $v$ in path $p$}
      \State \Call{Replace}{$p$, $v$, $a$} \label{repl1}
            \Comment{remove a Steiner vertex from $B$}
    \Else
      \State let $\{u,v\}$ be an edge between to terminal vertices of $B$
      \State \Call{Insert}{$p$, $v$, $a$} \label{ins1}
            \Comment{add terminal vertices of $A$ to the path}
    \EndIf
  \EndWhile
  \State \textbf{return} $\{ p \}$
\EndIf
%\EndFunction
\end{algorithmic}
\normalsize
%\hrule
\caption{{\sc SeriesComp}(set $A$, Cover $C_A$, set $B$, Cover $C_B$)}  
\label{fig:algorithm3xa}
\end{algorithm}

Similar to function {\sc SeriesComp} we introduce a function
{\sc OrderComp}, see algorithm \ref{fig:algorithm3xb},
which returns an optimal directed Steiner path cover
for some directed co-graph $G = A \oslash B$, if some optimal covers
for $A$ and $B$ are given as parameter.
The function {\sc Concat}$(p,p')$ used in line \ref{concat} extends path
$p$ by adding path $p'$ at its end.

\begin{algorithm}[ht]
\footnotesize
\begin{algorithmic}[1]
%\Function{OrderComp}{set $A$, Cover $C_A$, set $B$, Cover $C_B$}
\State let $C'_A := C_A$; let $C'_B := C_B$;
\If {$|C_A| < |C_B|$}
    \Loop ~ $\min\{ \s(C_A), |C_B| - |C_A| \}$-times
        \State let $r \in C'_A$ that contains a Steiner vertex
        \State $C'_A := C'_A - \{r\}$
        \State $C'_A := C'_A \cup \{r_1, r_2\}$ where $r_1, r_2$ result
             from  $r$ by splitting $r$ at a Steiner vertex
    \EndLoop
\EndIf
\If {$|C_A| > |C_B|$}
    \Loop ~ $\min\{ \s(C_B), |C_A| - |C_B| \}$-times
        \State let $r \in C'_B$ that contains a Steiner vertex
        \State $C'_B := C'_B - \{r\}$
        \State $C'_B := C'_B \cup \{r_1, r_2\}$ where $r_1, r_2$ result
             from  $r$ by splitting $r$ at a Steiner vertex
    \EndLoop
\EndIf
\State $C := \emptyset$
\Loop ~ $\min\{|C'_A|, |C'_B|\}$-times
    \State let $a \in C'_A$; $C'_A := C'_A - \{a\}$;
    \State let $b \in C'_B$; $C'_B := C'_B - \{b\}$;
    \State \Call{Concat}{$a,b$}  \label{concat}
    \State $C := C \cup \{a\}$
\EndLoop
\While {$C'_A \neq \emptyset$}
    \State let $a \in C'_A$; $C'_A := C'_A - \{a\}$; $C := C \cup \{a\}$
\EndWhile
\While {$C'_B \neq \emptyset$}
    \State let $b \in C'_B$; $C'_B := C'_B - \{b\}$; $C := C \cup \{b\}$
\EndWhile
\State \textbf{return} $C$
%\EndFunction
\end{algorithmic}
\caption{{\sc OrderComp}(set $A$, Cover $C_A$, set $B$, Cover $C_B$)}  
\label{fig:algorithm3xb}
\end{algorithm}

Let $G$ be a directed co-graph represented by its binary directed
co-tree $T(G)$. The function {\sc DirectedSteinerPathCover},
see algorithm \ref{fig:algorithm3x},
recursively computes a directed Steiner path cover of minimum cost
of the subgraph of $G$ induced by the vertices of $T(G)$. For the
series composition we assume that the left subtree $T_{x_\ell}$
of any vertex $x$ of $T(G)$ contains no more terminal vertices than
its right subtree $T_{x_r}$.
Otherwise we only had to swap the children of the vertex $x$.

\begin{algorithm}[ht]
\footnotesize
\begin{algorithmic}
%\Function{DirectedSteinerPathCover}{Co-Tree $T$, Vertex $x$}
\If {$x$ is the only vertex of $T$}
  \If {$x$ is a terminal vertex of $G$}
    \State \textbf{return} $\{(x)\}$
  \EndIf
  \State \textbf{return} $\emptyset$

\Else 
  \State $C_A$ := \Call{DirectedSteinerPathCover}{$T_{x_{\ell}}$,$x_{\ell}$}
       \Comment{$x_{\ell}$ is the left successor of $x$}
  \State $C_B$ := \Call{DirectedSteinerPathCover}{$T_{x_r}$,$x_r$}
       \Comment{$x_{r}$ is the right successor of $x$}
  
  \If {$x$  corresponds to a disjoint union  of $T$}
    \State \textbf{return}  $C_A \cup C_B$
   \EndIf                   
   \If {$x$  corresponds to a series composition of $T$}
    \State \textbf{return}
       \Call{SeriesComp}{$V(T_{x_{\ell}})$, $C_A$, $V(T_{x_{r}})$, $C_B$}
   \EndIf   
   \If {$x$  corresponds to an order composition of $T$}
    \State \textbf{return}
       \Call{OrderComp}{$V(T_{x_{\ell}})$, $C_A$, $V(T_{x_{r}})$, $C_B$}
   \EndIf   
\EndIf
%\EndFunction
\end{algorithmic}
\normalsize
%\hrule
\caption{{\sc DirectedSteinerPathCover}(Co-Tree $T$, Vertex $x$)}  
\label{fig:algorithm3x}
\end{algorithm}

By algorithm {\sc DirecedSteinerPathCover} we obtain the following result.

\begin{theorem} \label{thSteinerPathCover}
A  directed Steiner path cover of minimum cost for a directed
co-graph can be computed in linear time with respect to the size of the
directed co-expression.
\end{theorem}

\begin{proof}
The correctness follows by Lemma \ref{le-d-p} and \ref{le-d-s}.
The running time can be achieved by storing the paths
using double-linked, linear lists, where the paths that
contain Steiner vertices are stored in one set and the paths that
contain no Steiner vertices are stored in another set. The lists each
have a pointer to the first and last element, which are terminal
vertices, and they have a pointer to the first and last Steiner vertices.
Additionally, we store the number of terminal and Steiner vertices for
each list.
Each of the operations in Algorithm {\sc DirectedSteinerPathCover}
can be done in constant time.
\end{proof}

%%%%%%%%%%%%%%%%%%%%%%%%%%%%%%%%%%%%%%%%%%%%%%%%%%%%%%%%%%%%%%%%%%%%%%%%%%
\section{Hamiltonian Path Problem} \label{sec_hp}
%%%%%%%%%%%%%%%%%%%%%%%%%%%%%%%%%%%%%%%%%%%%%%%%%%%%%%%%%%%%%%%%%%%%%%%%%%

Our motivation to study the Hamiltonian Path problem comes from the
problem of minimizing change-over times in pick-and-place machines in
electronic industry. Before we study this problem, we recall the definition
of Hamiltonian Path problem.

\begin{desctight}
\item[Name] Shortest Directed Hamiltonian Path

\item[Instance] A directed graph $G$ and edge weights
  $c:  E(G) \to \RR^{\geq 0}$.

\item[Task] Find a directed Hamiltonian path $P$ in $G$ that minimizes
  $\sum_{e\in P} c(e)$.
\end{desctight}

Pick-and-place machines have been studied for many years, see \cite{CKS02}.
Often, the focus of planning problems in printed circuit board assembly
is on optimizing the throughput of these pick-and-place machines
to produce many printed circuit boards (short PCB) a day. We are
interested in minimizing change-over times for refilling the machines.

Mid-sized companies in the electronic industry often have to produce
different PCBs on one day. Each of these PCBs has
to be produced in small quantities. The SMD components like capacitors,
resistors, or integrated circuits are usually positioned on the boards
by pick-and-place machines and are soldered in a re-flow oven. Different
types of components are fed in to the machine by feeders or trays.
After one type of boards has been assembled, the pick-and-place
machine has to be reassembled with other types of SMD components. The
time for this refilling of components is often limited by the time the
last PCBs of the previous series stay in the re-flow oven.
The change-over must not take too long to avoid an unnecessary
downtime of the machine. Therefore, PCBs must be processed
in an order that only a few component groups have to be replaced.
We show how a graph model can be used to describe the problem, and we
give an integer program to solve the problem.

Assume that a company produces different PCBs, and the available
pick-and-place machine can be equipped with at most $k$ different types of
SMD components, since no more trays are available. At one day,
$n$ different PCBs have to be produced, that together contain
$m$ different types of SMD components $t_1,\ldots,t_m$. Each board $b_i$
has to be equipped with a set $T_i \subset \{t_1,\ldots,t_m\}$ of component
types. We represent each board by an $m$-tuple such that the $i$-th
entry is one if and only if type $t_i$ is needed to produce the board.
The number of ones in such a tuple is at most $k$.
For example, $b_i = (0,1,0,1,0)$ means that components of type 2 and 4
must be used, types 1, 3, and 5 are not used for producing this board.

In our graph model, each board is represented by a vertex, and two vertices
are connected by an edge. To change over from a board $u$ to a board $v$
those types have to be removed from the machine, that are used by $u$ but
not used by $v$. Those types, that are not used by $u$ but are used by $v$
have to be inserted in to the machine. Thus, the cost of change over is
described by the hamming distance of $u$ and $v$.

\begin{example}
Given $n=4$ jobs $b_1=(1,0,0,0)$, $b_2=(0,1,0,1)$, $b_3=(1,0,1,0)$, and
$b_4=(1,1,1,0)$. The resulting graph contains the artificial vertex
$b_0=(0,0,0,0)$ to represent the initial state of the machine, where all
trays are empty. It is shown in Figure~\ref{f_ex} together with the edge
costs. To minimize the overall change-over times is to produce the boards
in the order $b_1=(1,0,0,0)$, $b_3=(1,0,1,0)$, $b_4=(1,1,1,0)$, and
$b_2=(0,1,0,1)$ which sum up to 6.
\end{example}

\begin{figure}[hbtp]
\centerline{\epsfxsize=.6\textwidth \epsfbox{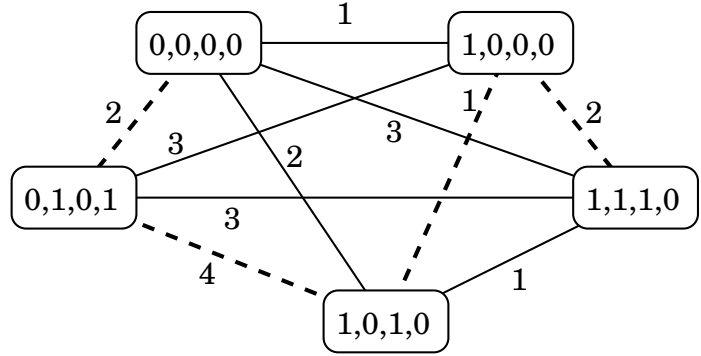}}
\caption{Graph modeling 4 boards. The dashed lines show a Hamiltonian path
  that defines an order in which the change-over times sum up to 9.}
\label{f_ex}
\end{figure}

A shortest Hamiltonian path that starts at the artificial vertex representing
the initial state of the machine minimizes the overall change-over time. In
the next section we give a mixed integer program to solve the problem.

%%%%%%%%%%%%%%%%%%%%%%%%%%%%%%%%%%%%%%%%%%%%%%%%%%%%%%%%%%%%%%%%%%%%%%%%%%
\subsection{Integer Program}
%%%%%%%%%%%%%%%%%%%%%%%%%%%%%%%%%%%%%%%%%%%%%%%%%%%%%%%%%%%%%%%%%%%%%%%%%%

In the following we assume that a directed graph $G=(V,E,c)$ is given,
where $V$ denotes the set of vertices, $E$ denotes the set of edges, and
$c$ denotes a function that assigns the cost $c(e)$ to edge $e$. We
represent an undirected graph $G'$ by a directed graph $G$ in such a way
that each edge $e$ of $G'$ is represented by two anti-parallel edges in $G$,
and the weight of the new edges are equal to the weight of edge $e$.

Our integer program is based on the following idea. Choose a vertex as
the first vertex on the path. This is typically an artificial vertex
that represents the empty machine at the beginning of a day. The order
of vertices on the path implies for two vertices $u$ and $v$ following
immediately on the path, that the cost of edge $(u,v)$ represents
the change-over time.

Our integer program uses binary variables $x[i,j]$, $i,j \in [n]$, such
that $x[i,j]$ is equal to one, if and only if edge $(i,j)$ is on the
Hamiltonian path. We use integer valued variables $p[i]$, $i \in [n]$,
such that $p[i] = \ell$, if and only if vertex $i$ is the $\ell$-th
vertex on the Hamiltonian path.

The first constraints ensure that the paths are all vertex-disjoint.
The path leaves each vertex $i$ at most once, and the path enters each
vertex $j$ at most once, which can be stated as follows.
\[
    \sum_{(i,j) \in E} x[i,j] \leq 1 ~~ \forall i \in V
    \quad \quad \mbox{and} \quad \quad
    \sum_{(i,j) \in E} x[i,j] \leq 1 ~~ \forall j \in V
\]
The following constraints prevent loops over anti-parallel edges. If there are anti-parallel
edges then choose at most one of them. This is ensured by the constraints:
\[
    x[i,j] \leq 1 - x[j,i] ~~ \forall (i,j) \in E: (j,i) \in E
\]
Each Hamiltonian path has length $n-1$:
\[
    \sum_{(i,j) \in E} x[i,j] = n-1
\]
With just this constraint on the number of edges, we may accept
solutions that are not connected. It is still possible to have cycles
and a path instead of one path only. The next constraints are used to
avoid cycles. If edge $(i,j)$ belongs to the path then $p[j] - p[i] = 1$
must hold true:
\[
\begin{array}{crcl}
    & -n  + (n+1) \cdot x[i,j] & \leq & p[j] - p[i] ~~ \forall (i,j) \in E \\
  \mbox{and} &
    p[j] - p[i] & \leq & n - (n-1) \cdot x[i,j] ~~ \forall (i,j) \in E
\end{array}
\]
The objective is to minimize the sum of the change-over times, i.e.\ to
make the path length as small as possible, so we have to minimize the
objective function $\sum_{(i,j) \in E} c[i,j] \cdot x[i,j]$.
Alternatively, we can use $\max_{(i,j) \in E} c[i,j] \cdot x[i,j]$ as
objective function to minimize the largest change-over time.

%%%%%%%%%%%%%%%%%%%%%%%%%%%%%%%%%%%%%%%%%%%%%%%%%%%%%%%%%%%%%%%%%%%%%%%%%%
\subsection{Experimental results}
%%%%%%%%%%%%%%%%%%%%%%%%%%%%%%%%%%%%%%%%%%%%%%%%%%%%%%%%%%%%%%%%%%%%%%%%%%

We ran the mixed integer program on CPLEX 12.8.0 using default settings
on some standard PC with Intel i7, 2.80 GHz CPU and 16 GB RAM running
Ubuntu 18.10.

Running times to solve the Hamiltonian Path problem on randomly generated
tournaments with randomly generated edge weights are shown in
Table \ref{tab_erg_HP}. The running times are average values over 10 runs.
Bold numbers represent times when the time limit of 1800 seconds, which
is equal to 30 minutes, was reached at least once and at most twice.
It is well-known that any tournament on a finite number of vertices
contains a Hamiltonian path.

\begin{table}[hbtp]
\[
  \begin{array}{|r||r|r|r|r|r|}
    \hline
    n  & 100 & 200 & 300 & 400 & 500 \\
    \hline \hline
    \text{time}  & 2{.}1 & 22{.}1 & 69{.}3 & {\bf 519} & {\bf 783} \\
    \text{stdev} & 0{.}7 & 19{.}1 & 81{.}6 & 741 & 885 \\
    \hline
  \end{array}
\]
\caption{Running times of MIP to calculate Hamiltonian Path on randomly
  generated tournaments.}
\label{tab_erg_HP}
\end{table}

Running times to solve the Hamiltonian Path problem on complete bipartite
graphs $K_{n,m}$ for different values of $n$ and $m$ are shown in Table
\ref{tab_erg_HP2}. The complete
bipartite graph $K_{n,m}$ contains a Hamiltonian path if and only if
$|n-m| \leq 1$. The absence of such a path can be detected very quickly
by the MIP as can be seen in the columns $m=n+2$.
All complete bipartite graphs are co-graphs so these results show that
our algorithm with running time linear in the input size would be very
much faster than the integer program.
Running times of more than 30 minutes are indicated by a bar $-$.

\begin{table}[hbtp]
\[
  \begin{array}{|r||r|r|r||r|r|r||r|r|r||r|r|r||r|r|r|}
    \hline
     & \multicolumn{3}{c||}{n = 50} & \multicolumn{3}{c||}{n = 100} &
       \multicolumn{3}{c||}{n = 150} & \multicolumn{3}{c||}{n = 200} &
       \multicolumn{3}{c|}{n = 250} \\
   \hline
   m & 50 & 51 & 52 & 100 & 101 & 102 & 150 & 151 & 152 &
      200 & 201 & 202 & 250 & 251 & 252 \\
   \hline \hline
   time & 1{.}7 & 3{.}2 & 0{.}1 & 16{.}7 & 10{.}9 & 0{.}1 &
       131 & 366 & 0{.}2 & - & 370 & 0{.}3 & - & 1774 & 0{.}5 \\
   \hline
  \end{array}
\]
\caption{Running times of MIP to calculate Hamiltonian Path on complete
  bipartite graphs $K_{n,m}$ for different values of $n$ and $m$.}
\label{tab_erg_HP2}
\end{table}

%%%%%%%%%%%%%%%%%%%%%%%%%%%%%%%%%%%%%%%%%%%%%%%%%%%%%%%%%%%%%%%%%%%%%%%%%%
%%%%%%%%%%%%%%%%%%%%%%%%%%%%%%%%%%%%%%%%%%%%%%%%%%%%%%%%%%%%%%%%%%%%%%%%%%
\section{Steiner Path Problem}
%%%%%%%%%%%%%%%%%%%%%%%%%%%%%%%%%%%%%%%%%%%%%%%%%%%%%%%%%%%%%%%%%%%%%%%%%%
%%%%%%%%%%%%%%%%%%%%%%%%%%%%%%%%%%%%%%%%%%%%%%%%%%%%%%%%%%%%%%%%%%%%%%%%%%

% Motivation for Steiner Path ...

%%%%%%%%%%%%%%%%%%%%%%%%%%%%%%%%%%%%%%%%%%%%%%%%%%%%%%%%%%%%%%%%%%%%%%%%%%
\subsection{Integer Program}
%%%%%%%%%%%%%%%%%%%%%%%%%%%%%%%%%%%%%%%%%%%%%%%%%%%%%%%%%%%%%%%%%%%%%%%%%%

Here again, we assume that a directed graph $G=(V,E,c)$ is given,
where $V$ denotes the set of vertices, $E$ denotes the set of edges, and
$c$ denotes a function that assigns the cost $c(e)$ to edge $e$. We
represent an undirected graph $G'$ by a directed graph $G$ in such a way
that each edge $e$ of $G'$ is represented by two anti-parallel edges in $G$,
and the weight of the new edges are equal to the weight of edge $e$.

We use additional binary variables $y[i]$, $i \in [n]$, such that
$y[i]$ is equal to one, if and only if vertex $i$ is contained in the
path. These variables are only used to simplify the formulation of
the following constraints \ref{constr_1} and \ref{constr_2}.
As in the last chapter, our integer program uses
binary variables $x[i,j]$, $i,j \in [n]$, such that $x[i,j]$ is equal
to one, if and only if edge $\{i,j\}$ is on the Steiner path.
We use integer valued variables $p[i]$, $i \in [n]$, such that
$p[i] = \ell$, if and only if vertex $i$ is the $\ell$-th
vertex on the Steiner path.
Each terminal vertex $i$ has to be contained in the path.
\[
    p[i] \geq 1 ~~ \forall i \in T
\]
If $x[i,j] = 1$ then $i$ and $j$ lie on the path, i.e.\ if $x[i,j] = 1$
then $p[i] \neq 0$ and $p[j] \neq 0$ which can be stated by the following
constraints.
\[
    x[i,j] \leq p[i] ~~ \forall (i,j) \in E
    \quad \quad \mbox{ and } \quad \quad
    x[i,j] \leq p[j] ~~ \forall (i,j) \in E
\]
If there are anti-parallel edges then choose at most one of them. This is
ensured by the constraint:
\[
    x[i,j] \leq 1 - x[j,i] ~~ \forall (i,j) \in E: (j,i) \in E
\]
It holds $y[i] = 1$ iff $p[i] > 0$, i.e. $y[i] = 1$ iff vertex $i$ is on
the path.
\[
   y[i] \leq p[i] ~~ \forall i \in V
    \quad \quad \mbox{ and } \quad \quad
   p[i] \leq n \cdot y[i] ~~ \forall i \in V
\]
The path has to contain one edge less than vertices.
\begin{eqnarray}
    \sum_{(i,j) \in E} x[i,j] & = & \sum_{i \in V} y[i] - 1 \label{constr_1}
\end{eqnarray}
If $p[i] \neq 0$ then there has to be a $x[i,j] = 1$ or $x[j,i] = 1$.
\[
    p[i] \leq n \cdot \left(\sum_{(i,j) \in E} x[i,j] +
                            \sum_{(j,i) \in E} x[j,i]\right)
    ~~ \forall i \in V
\]
For each vertex $i \in V-T$ contained in the path we have one in-going
and one outgoing edge.
\begin{eqnarray}
  \sum_{(i,j) \in E} x[i,j] = y[i] ~~ \forall i \in V-T
    & \mbox{ and } &
  \sum_{(i,j) \in E} x[i,j] = y[i] ~~ \forall j \in V-T \label{constr_2}
\end{eqnarray}
The path leaves each vertex $i \in T$ at most once and it enters each
vertex $j \in T$ at most once.
\[
    \sum_{(i,j) \in E} x[i,j] \leq 1 ~~ \forall i \in T
    \quad \quad \mbox{ and } \quad \quad
    \sum_{(i,j) \in E} x[i,j] \leq 1 ~~ \forall j \in T
\]
If a vertex is in the path then its position must be at least one.
\[
    \sum_{(i,j) \in E} x[i,j] \leq p[i] ~~ \forall i \in V
    \quad \quad \mbox{ and } \quad \quad
    \sum_{(i,j) \in E} x[i,j] \leq p[j] ~~ \forall j \in V
\]
Finally, to avoid cycles, if edge $(i,j)$ belongs to the path then
$p[j] - p[i] = 1$ must hold true:
\[
\begin{array}{crcl}
    & -n  + (n+1) \cdot x[i,j] & \leq & p[j] - p[i] ~~ \forall (i,j) \in E \\
  \mbox{and} &
    p[j] - p[i] & \leq & n - (n-1) \cdot x[i,j] ~~ \forall (i,j) \in E
\end{array}
\]
The objective is to minimize the sum of edge-weights
$\sum_{(i,j) \in E} c[i,j] \cdot x[i,j]$. For unit-distance graphs
this objective function results in a Steiner path with least number
of Steiner vertices.

%%%%%%%%%%%%%%%%%%%%%%%%%%%%%%%%%%%%%%%%%%%%%%%%%%%%%%%%%%%%%%%%%%%%%%%%%%
\subsection{Experimental results}
%%%%%%%%%%%%%%%%%%%%%%%%%%%%%%%%%%%%%%%%%%%%%%%%%%%%%%%%%%%%%%%%%%%%%%%%%%

Running times to solve Steiner Path problems on complete bipartite graphs
$K_{n,3n}$ for different values of $n$ and $t$ many randomly chosen terminal
vertices are shown in Table \ref{tab_erg_StP}. As long as at most $n$
vertices from the second set are selected
as terminal vertices, a Steiner path exists.
The absence of such a path can be detected very quickly
by the MIP as can be seen in the columns $t=2n$.

\begin{table}[hbtp]
\[
  \begin{array}{|r||r|r|r||r|r|r||r|r|r||r|r|r||r|r|r|}
    \hline
     & \multicolumn{3}{c||}{n = 25} & \multicolumn{3}{c||}{n = 50} &
       \multicolumn{3}{c||}{n = 75} & \multicolumn{3}{c||}{n = 100} &
       \multicolumn{3}{c|}{n = 125} \\
   \hline
   t & 12 & 25 & 50 & 25 & 50 & 100 & 37 & 75 & 150 &
       50 & 100 & 200 & 62 & 125 & 250 \\
   \hline \hline
   \text{time} & 0{.}9 & 1{.}2 & 0{.}1 & 7{.}2 & 15{.}5 & 0{.}2 &
       29{.}7 & 139 & 0{.}4 & 189 & {\bf 945} & 0{.}7 &
       {\bf 643} & - & 1{.}1 \\
   \hline
  \end{array}
\]
\caption{Running times to solve Steiner Path problems on complete
  bipartite graphs $K_{n,3n}$ for different values of $n$ and $t$}
\label{tab_erg_StP}
\end{table}

The running times are average values over 10 runs,
bold numbers represent times when the time limit of 1800 seconds, which
is equal to 30 minutes, was reached at least once and at most twice.
%times printed in bold font
%contain at most 2 runs where the time limit of 1800 seconds was reached.
Running times of more than 30 minutes are indicated by a bar $-$.

%%%%%%%%%%%%%%%%%%%%%%%%%%%%%%%%%%%%%%%%%%%%%%%%%%%%%%%%%%%%%%%%%%%%%%%%%%
%%%%%%%%%%%%%%%%%%%%%%%%%%%%%%%%%%%%%%%%%%%%%%%%%%%%%%%%%%%%%%%%%%%%%%%%%%
\section{Steiner Path Cover Problem}
%%%%%%%%%%%%%%%%%%%%%%%%%%%%%%%%%%%%%%%%%%%%%%%%%%%%%%%%%%%%%%%%%%%%%%%%%%
%%%%%%%%%%%%%%%%%%%%%%%%%%%%%%%%%%%%%%%%%%%%%%%%%%%%%%%%%%%%%%%%%%%%%%%%%%

% Motivation for Steiner Path Cover ...

%%%%%%%%%%%%%%%%%%%%%%%%%%%%%%%%%%%%%%%%%%%%%%%%%%%%%%%%%%%%%%%%%%%%%%%%%%
\subsection{Integer Program}
%%%%%%%%%%%%%%%%%%%%%%%%%%%%%%%%%%%%%%%%%%%%%%%%%%%%%%%%%%%%%%%%%%%%%%%%%%

To simplify the objective function and the IP we
modify the given directed graph $G=(V,E,c)$ in such a way,
that a new vertex $s$ as source and another new vertex $t \neq s$
as sink is added. Vertex $s$ gets edges to each terminal vertex
and each terminal vertex gets an edge to the sink $t$ as
shown in Figure~\ref{f_ex2}.
\begin{figure}
\centerline{\epsfxsize=.5\textwidth \epsfbox{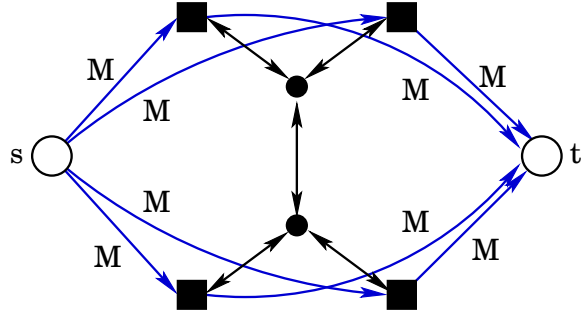}}
\caption{Graph extended by a source and a sink to compute an optimal
  Steiner Path Cover.}
\label{f_ex2}
\end{figure}
The additional edges get a weight of $M := |V| \cdot |E|$ if the graph
has no edge weights, otherwise $M := |V| \cdot C$ where
$C := \max_{e \in E}\{ c(e) \}$ denotes the maximum edge
weight of the graph. More formally, we get graph $G'=(V',E',c')$
such that $V' := V \cup \{s,t\}$, $E' := E \cup \{(s,v) \mid v \in T\}
\cup \{(v,t) \mid v \in T\}$, and $c'(e) = c(e)$ for each $e \in E$ and
$c'(e) = M$ if $e=(s,v)$ or $c=(v,t)$ for some $v \in T$. If graph $G$
has no edge weights then $c'(e) = 1$ for each $e \in E$, if $G$ is an
undirected graph we replace each edge $\{u,v\} \in E$ by two anti-parallel
edges $(u,v)$ and $(v,u)$.

As in the last chapter, our integer program uses binary variables
$y[i]$, $i \in [n]$, such that $y[i]$ is equal to one, if and only
if vertex $i$ is contained in the path. Binary variables
$x[i,j]$, $i,j \in [n]$, are used such that $x[i,j]$ is equal
to one, if and only if edge $\{i,j\}$ is on the Steiner path.
We use integer valued variables $p[i]$, $i \in [n]$, such that
$p[i] = \ell$, if and only if vertex $i$ is the $\ell$-th
vertex on the Steiner path.

The number of paths in the cover is equal to $\sum_{(s,v) \in E'} x[s,v]$,
and since there is at least one path in a Steiner Path Cover for $G$
if $T \neq \emptyset$, the following conditions must hold true:
\[
    \sum_{(s,v) \in E'} x[s,v] > 0
    \quad \quad \mbox{ and } \quad \quad
    \sum_{(v,t) \in E'} x[v,t] > 0
\]
Because we have inserted artificial vertices, the path leaves each
terminal vertex $i$ exactly once, and the path enters each terminal
vertex $j$ exactly once, which can be stated as follows.
\[
    \sum_{(i,j) \in E'} x[i,j] = 1 ~~ \forall i \in T
    \quad \quad \mbox{and} \quad \quad
    \sum_{(i,j) \in E'} x[i,j] = 1 ~~ \forall j \in T
\]
Since non-terminal vertices need not be contained in the path, the
path leaves each non-terminal vertex $i$ at most once, and the path enters
each non-terminal vertex $j$ at most once, which can be stated as follows.
\[
    \sum_{(i,j) \in E} x[i,j] \leq 1 ~~ \forall i \in V-T
    \quad \quad \mbox{and} \quad \quad
    \sum_{(i,j) \in E} x[i,j] \leq 1 ~~ \forall j \in V-T
\]
The previous conditions also ensure that paths are vertex-disjoint.
If a path contains a non-terminal vertex $j$ then there has to be an
edge that enters $j$ and an edge that leaves $j$, which can be stated as:
\[
    \sum_{(i,j) \in E} x[i,j]  =  \sum_{(j,k) \in E} x[j,k]
     ~~ \forall j \in V-T
\]
If there are anti-parallel edges then choose at most one of them. This is
ensured by the constraint:
\[
    x[i,j] \leq 1 - x[j,i] ~~ \forall (i,j) \in E: (j,i) \in E
\]
Each terminal vertex $i$ has to be contained in the path.
\[
    p[i] \geq 1 ~~ \forall i \in T
\]
If $x[i,j] = 1$ then $i$ and $j$ lie on the path, i.e.\ if $x[i,j] = 1$
then $p[i] \neq 0$ and $p[j] \neq 0$ which can be stated by the following
constraints.
\[
    x[i,j] \leq p[i] ~~ \forall (i,j) \in E
    \quad \quad \mbox{ and } \quad \quad
    x[i,j] \leq p[j] ~~ \forall (i,j) \in E
\]
If $p[i] \neq 0$ then there has to be a $x[i,j] = 1$ or $x[j,i] = 1$.
\[
    p[i] \leq n \cdot \left(\sum_{(i,j) \in E} x[i,j] +
                            \sum_{(j,i) \in E} x[j,i]\right)
    ~~ \forall i \in V
\]
Finally, to avoid cycles, if edge $(i,j)$ belongs to the path then
$p[j] - p[i] = 1$ must hold true:
\[
\begin{array}{crcl}
    & -n  + (n+1) \cdot x[i,j] & \leq & p[j] - p[i] ~~ \forall (i,j) \in E \\
  \mbox{and} &
    p[j] - p[i] & \leq & n - (n-1) \cdot x[i,j] ~~ \forall (i,j) \in E
\end{array}
\]

Minimizing an objective function like $\sum_{v \in T} x[s,v]$ would only
minimize the number of paths. If there are two different path covers of
minimum size, we have to select that one with the least number of Steiner
vertices. By our definition of the costs of the additional edges from $s$ to
terminal vertices we can choose the same objective function as above, when
we also take care of the additional edges.
The objective is to minimize the sum of edge-weights
$\sum_{(i,j) \in E'} c[i,j] \cdot x[i,j]$. For unit-distance graphs this
objective function results in a Steiner path cover $P$ of minimum size,
with least number of Steiner vertices $\sum_{p \in P} \sum_{e \in p} c(e)$.

%%%%%%%%%%%%%%%%%%%%%%%%%%%%%%%%%%%%%%%%%%%%%%%%%%%%%%%%%%%%%%%%%%%%%%%%%%
\subsection{Experimental results}
%%%%%%%%%%%%%%%%%%%%%%%%%%%%%%%%%%%%%%%%%%%%%%%%%%%%%%%%%%%%%%%%%%%%%%%%%%

Running times to solve Steiner Path Cover problems on complete bipartite
graphs $K_{n,3n}$ for different values of $n$ and $t$ many randomly
chosen terminal vertices are shown in Table \ref{tab_erg_SPC}.

\begin{table}[hbtp]
\[
  \begin{array}{|r||r|r|r||r|r|r||r|r|r||r|r|r||r|r|r|}
   \hline
     & \multicolumn{3}{c||}{n = 25} & \multicolumn{3}{c||}{n = 50} &
       \multicolumn{3}{c||}{n = 75} & \multicolumn{3}{c||}{n = 100} &
       \multicolumn{3}{c|}{n = 125} \\
   \hline
   t & 12 & 25 & 50 & 25 & 50 & 100 & 37 & 75 & 150 &
       50 & 100 & 200 & 62 & 125 & 250 \\
   \hline \hline
   \text{time} & 0{.}9 & 1{.}4 & 1{.}1 & 7{.}4 & 9{.}1 & 7{.}4 &
       25{.}7 & 47{.}7 & 20{.}9 & 112 & 207 & 55{.}8 & 307 & - & 269 \\
   \hline
  \end{array}
\]
\caption{Running times to solve Steiner Path Cover problems on complete
  bipartite graphs $K_{n,3n}$ for different values of $n$ and $t$.}
\label{tab_erg_SPC}
\end{table}

The running times are average values over 10 runs.
Running times of more than 30 minutes are indicated by a bar $-$,
bold numbers represent times when the time limit
was reached at least once and at most twice.
Table \ref{tab_erg_SPC2} shows
running times to solve Steiner Path Cover problems on
randomly generated co-graphs with $n$ vertices and $t$
many randomly chosen terminal vertices.

\begin{table}[hbtp]
\[
  \begin{array}{|r||r|r|r||r|r|r||r|r|r||r|r|r|}
   \hline
     & \multicolumn{3}{c||}{n = 200} & \multicolumn{3}{c||}{n = 300} &
       \multicolumn{3}{c||}{n = 400} & \multicolumn{3}{c|}{n = 500} \\
   \hline
   t & n/5 & n/2 & 4n/5 & n/5 & n/2 & 4n/5 &
       n/5 & n/2 & 4n/5 & n/5 & n/2 & 4n/5 \\
   \hline \hline
   \text{time} & 9{.}3 & 13{.}4 & 19{.}7 & 21{.}8 & 72{.}9 & {\bf 337} &
         46{.}1 & 191 & {\bf 413} & 121 & {\bf 524} & {\bf 708} \\
   \hline
  \end{array}
\]
\caption{Running times to solve Steiner Path Cover problems on randomly
  generated co-graphs with $n$ vertices and $t$ many randomly chosen
  terminal vertices.}
\label{tab_erg_SPC2}
\end{table}

The co-graphs will be generated recursively as shown in Algorithm
\ref{algo_randCoGr}. The probability $p$ is set to be $\nicefrac{1}{3}$
in our experiments.

\begin{algorithm}[ht]
%\small
\begin{algorithmic}
\State choose $t$ distinct vertices at random
\State call \Call{otimes}{1, n}
\Function{otimes}{int l, int r}
  \If {r - l = 1}
    \State create edge $(l,r)$ and \textbf{return}
  \EndIf
  \State m := rand(l, r) \Comment{choose random number between l and r}
  \State call \Call{oplus}{l, m} with probability $p$, otherwise
    call \Call{otimes}{l, m}
  \State call \Call{oplus}{m+1, r} with probability $p$, otherwise
    call \Call{otimes}{m+1, r}
  \State create edges $(i,j)$ for all vertices $l \leq i \leq m$, and
    $m+1 \leq j \leq r$
\EndFunction
\Function{oplus}{int l, int r}
  \If {r - l = 1}
    \State \textbf{return}
  \EndIf
  \State m := rand(l, r)
  \State call \Call{oplus}{l, m} with probability $p$, otherwise
   call \Call{otimes}{l, m}
  \State call \Call{oplus}{m+1, r} with probability $p$, otherwise
    call \Call{otimes}{m+1, r}
\EndFunction
\end{algorithmic}
%\normalsize
%\hrule
\caption{{\sc Random-Co-Graph}(int t, int n)}
\label{algo_randCoGr}
\end{algorithm}

%%%%%%%%%%%%%%%%%%%%%%%%%%%%%%%%%%%%%%%%%%%%%%%%%%%%%%%%%%%%%%%%%%%%%%%%%%
%%%%%%%%%%%%%%%%%%%%%%%%%%%%%%%%%%%%%%%%%%%%%%%%%%%%%%%%%%%%%%%%%%%%%%%%%%
\section{Conclusions and outlook}
%%%%%%%%%%%%%%%%%%%%%%%%%%%%%%%%%%%%%%%%%%%%%%%%%%%%%%%%%%%%%%%%%%%%%%%%%%
%%%%%%%%%%%%%%%%%%%%%%%%%%%%%%%%%%%%%%%%%%%%%%%%%%%%%%%%%%%%%%%%%%%%%%%%%%

In this paper we considered the directed Steiner path cover problem
for directed co-graphs. We could show a linear time solution for
computing the minimum number of paths within a directed Steiner path
cover and the minimum number of Steiner vertices in such a directed
Steiner path cover in directed co-graphs.  The results allowed us to 
give an algorithm which constructs a directed Steiner path cover of
minimum cost  for a directed co-graph in linear time. This leads to
a linear time computation of an optimal directed Steiner path, if it
exists, for directed co-graphs.

Undirected co-graphs are precisely those graphs which can be generated
from the single vertex graph by disjoint union and join operations,
see \cite{CLS81}.  Given some undirected co-graph $G$, we can solve the
Steiner path cover problem in linear time by replacing every edge 
$\{u,v\}$ of $G$  by two directed edges $(u,v)$ and $(v,u)$ and applying
our solution for directed co-graphs. This reproves our result of \cite{GHKRRW20a}.

The directed Hamiltonian path problem can be solved by an XP-algorithm
w.r.t.\ the parameter directed clique-width \cite{GHO13}. Since
directed co-graphs have directed clique-width at most two \cite{GWY16}
a polynomial time solution for the directed Hamiltonian path problem
follows. Such an algorithm is also given in \cite{Gur17a}.
Since a directed Hamiltonian path exists if and only if we have $T=V(G)$ and
$p(G)=1$, our results imply the first linear time algorithm for the 
directed Hamiltonian path problem on directed co-graphs. This
generalizes the known results for undirected co-graphs of
Lin et al.~\cite{LOP95}.

In our future work we want to find out whether results can be transferred
to other graph classes such as chordal graphs, interval graphs, or
distance-hereditary graphs.

%%%%%%%%%%%%%%%%%%%%%%%%%%%%%%%%%%%%%%%%%%%%%%%%%%%%%%%%%%%%%%%%%%%%%%
\section*{Acknowledgements} \label{sec-a}
%%%%%%%%%%%%%%%%%%%%%%%%%%%%%%%%%%%%%%%%%%%%%%%%%%%%%%%%%%%%%%%%%%%%%%

The work of the second and third author was supported
by the Deutsche
Forschungsgemeinschaft (DFG, German Research Foundation) -- 388221852

%%%%%%%%%%%%%%%%%%%%%%%%%%%%%%%%%%%%%%%%%%%%%%%%%%%%%%%%%%%%%%%%%%%%%%%%%%
%%%%%%%%%%%%%%%%%%%%%%%%%%%%%%%%%%%%%%%%%%%%%%%%%%%%%%%%%%%%%%%%%%%%%%%%%%

%\bibliographystyle{alpha}
%\bibliography{/home/gurski/bib.bib}

\newcommand{\etalchar}[1]{$^{#1}$}

\end{document}